\documentclass[12pt]{article}\usepackage[]{graphicx}\usepackage[]{xcolor}
\makeatletter
\def\maxwidth{ %
  \ifdim\Gin@nat@width>\linewidth
    \linewidth
  \else
    \Gin@nat@width
  \fi
}
\makeatother

\definecolor{fgcolor}{rgb}{0.345, 0.345, 0.345}

\usepackage{framed}
\makeatletter
 {\par\unskip\endMakeFramed%
 \at@end@of@kframe}
\makeatother

\definecolor{shadecolor}{rgb}{.97, .97, .97}
\definecolor{messagecolor}{rgb}{0, 0, 0}
\definecolor{warningcolor}{rgb}{1, 0, 1}
\definecolor{errorcolor}{rgb}{1, 0, 0}
\newenvironment{knitrout}{}{} 

\usepackage{alltt}

\usepackage{amsmath}
\usepackage{graphicx}
\usepackage{natbib}
\usepackage{url} 
\usepackage{amsfonts}
\usepackage{amsthm}
\usepackage{bm}
\usepackage{setspace}
\usepackage{caption}
\usepackage{subcaption}
\usepackage{placeins}
\usepackage[ruled,noline,linesnumbered]{algorithm2e}
\usepackage[table]{xcolor}
\usepackage{mathtools}
\usepackage{calc}
\usepackage{accents}
\usepackage{appendix}

\DeclareMathOperator*{\argmax}{arg\,max}
\DeclareMathOperator*{\argmin}{arg\,min}
\usepackage{etoolbox} 
\usepackage{enumitem}

\newcommand{\blind}{0}

\newcommand\R{\mathbb{R}}

\newcommand\unit{u} 
\newcommand\Unit{U} 

\newcommand\nmif{m} 

\newcommand\shared{\phi}
\newcommand\specific{\psi}
\newcommand\nshared{d_\shared}
\newcommand\nspecific{d_\specific}
\newcommand\BorelT{\mathcal{B}^{\nshared+\Unit\nspecific}}
\newcommand\regFun{\gamma}

\newcommand\iid{\text{iid}}

\newcommand\seq[2]{{#1}\!:\!{#2}}
\newcommand\prob{\mathbb{P}}
\newcommand\eulerstep{\Delta}
\newcommand\nclone{\tilde{n}}
\newcommand\nstack{n'}
\newcommand\addRV{Z}

\newcommand\lemmaBound{c}
\newcommand\normal{{\mathcal N}}
\newcommand\pop{\mathrm{pop}}

\newtheorem{lemma}{Lemma}
\newtheorem{theorem}{Theorem}
\newtheorem{corollary}{Corollary}

\newcommand{\dbtilde}[1]{\accentset{\approx}{#1}}

\IfFileExists{upquote.sty}{\usepackage{upquote}}{}
\begin{document}

\def\spacingset#1{\renewcommand{\baselinestretch}%
{#1}\small\normalsize} \spacingset{1}


\if0\blind
{
\title{\bf Iterating marginalized Bayes maps for likelihood maximization with application to nonlinear panel models}

\author{
  Jesse Wheeler\thanks{Email: jessewheeler@isu.edu, ORCiD: 0000-0003-3941-3884} \\
  Department of Mathematics and Statistics, Idaho State University
  \and
  Aaron J. Abkemeier \\
  Department of Statistics, University of Michigan
  \and
  Edward L. Ionides \\
  Department of Statistics, University of Michigan
}
  \maketitle
} \fi

\if1\blind
{
  \bigskip
  \bigskip
  \bigskip
  \begin{center}
    {\Large\bf  }
\end{center}
  \medskip
} \fi

\bigskip
\begin{abstract}
\noindent
Complex dynamic systems can be investigated by fitting mechanistic stochastic dynamic models to time series data.
In this context, commonly used Monte Carlo inference procedures for model selection and parameter estimation quickly become computationally unfeasible as the system dimension grows.
The increasing prevalence of panel data, characterized by multiple related time series, therefore necessitates the development of inference algorithms that are effective for this class of high-dimensional mechanistic models.
Nonlinear, non-Gaussian mechanistic models are routinely fitted to time series data but seldom to panel data, despite its widespread availability, suggesting that the practical difficulties for existing procedures are prohibitive. 
We investigate the use of iterated filtering algorithms for this purpose.
We introduce a novel algorithm that contains a marginalization step that mitigates issues arising from particle filtering in high dimensions. 
Our approach enables likelihood-based inference for models that were previously considered intractable, thus broadening the scope of dynamic models available for panel data analysis.
\end{abstract}

\noindent%
\vfill

\newpage
\spacingset{1.9}

\AtBeginEnvironment{algorithm}{%
  \singlespacing
}

\section{Introduction}
\label{sec:intro}

\noindent Panel data, otherwise known as longitudinal data, are are a collection of related time series.
Each time series is a measurement on a person, animal or object known as the {\it unit}.
A unit may have its treatment assigned via a randomized experiment, or it may be measured in an observational study.
The measurement on each unit at each observation time may be vector-valued.
While each time series within a panel dataset could be analyzed individually, there is an advantage to analyzing the entire collection of time series simultaneously.
For example, data from multiple measurement units during an infectious disease outbreak may reveal transmission dynamics not evident from individual units \citep{wiens14,ranjeva19}.
Other examples include ecological experiments and observational studies in which data are collected over time across multiple measurement sites that may experience variation in covariates of interest \citep{searle16,hewitt24}.
Modeling the data collectively allows for researchers to learn about processes that are shared across units, as well as identifying traits that are unique to each unit.

A common approach to modeling nonlinear dynamic systems is through the use of mechanistic models \citep{auger21}.
These models involve the proposal of a system of equations that describe how unobserved dynamic states evolve over time.
When combined with a model relating the latent states to observable quantities, we obtain a partially observed Markov process (POMP) model, also known as a state space model (SSM) or a hidden Markov model.
Once calibrated to data, these models provide a quantitative description of the observations while simultaneously adhering to a given scientific hypothesis about how the data are generated.
Despite the growing ability of deep learning and other advanced machine learning techniques to extract useful information from complex data sets, mechanistic models continue to be an important tool for modern science \citep{baker18,hogg24}.
In addition to improved interpretability, a key advantage of these models is that they enable the estimation of counter-factual scenarios, such as impact of interventions on a dynamic system \citep{wheeler24}.  

Various algorithms, both frequentist and Bayesian, are capable in principle of fitting mechanistic models to panel data.
Some algorithms scale well as the model's dimension grows, but they rely on unrealistic approximations \citep{evensen09,wigren22}, or they avoid the difficulties related to evaluating the model's likelihood function by optimizing alternative measures of goodness-of-fit or approximations to the likelihood \citep{toni09,wood10,whitehouse23,haggstrom25}.
Other algorithms can be applied in less restrictive cases, but they scale poorly as the model's dimension increases \citep{andrieu10,ionides15}.

Contemporary applications of nonlinear mechanistic models for low-dimensional systems are abundant \citep[e.g.,][]{kramer24,he24,newman23}, yet examples of higher-dimensional panel equivalents remain sparse.
The lack of examples indicates that scientists have struggled, largely unsuccessfully, to use existing statistical methodology for nonlinear panel models.
Our work addresses this possibility by proposing a new algorithm that is effective on benchmark problems, theoretically supported, and easy to implement using existing software such as the R package \texttt{panelPomp} \citep{panelpomp}, or a newly developed python package called \texttt{pypomp}, available on GitHub and \texttt{PyPI} \citep{pypomp25}.

An extensively used approach in ecology, known as \emph{data cloning} \citep{lele07}, involves iteratively using Bayes' rule, by recursively mapping a prior distribution to a posterior distribution using the same likelihood function, until the iterated posterior distribution converges to a point mass at the maximum likelihood estimate (MLE).
\citet{ionides15} developed an extension of data cloning for SSMs by treating the model parameters as latent states, and performing a random walk for these parameters at each observation time.
By iteratively applying a particle filter \citep{arulampalam02} to this extended model and decreasing the random walk standard deviations over time, it can be shown that the parameters will converge to the MLE after a sufficiently large number of iterations \citep{ionides15, chen25}.
This algorithm, known as IF2, has seen extensive use for inference on low dimensional dynamic models, especially in epidemiological contexts \citep[e.g.,][]{pons-salort18,stocks18,subramanian21,fox22}.

In panel models, the iterated filtering algorithm must be extended to handle the Monte Carlo error that grows exponentially with the number of units.
To address this, \citet{breto20} proposed an algorithm called the panel iterated filter (PIF).
As discussed in Appendix~\ref{sec:panelTheory}, this algorithm is a special case of iterated filtering where the model structure and random walk sequence have been modified to reduce the loss of information that is described by \citet{Liu01}.
The PIF algorithm has been effective in obtaining maximum likelihood estimates for highly nonlinear, non-Gaussian panel models in previous works \citep{ranjeva17,ranjeva19,wale19,lee20,domeyer22}.
However, as we demonstrate here, the PIF algorithm can be inefficient when the number of units is large, due to the resampling of all parameters at each step of the iterated particle filters.
In order to improve the computational inefficiencies of the PIF algorithm for high-dimensional panel models, we present a novel inference technique which we call the marginalized panel iterated filter (MPIF).
The algorithm reduces the number of times particles are resampled, which results in increased diversity of the particles that represent the parameter distribution at each iteration and time step.

The article proceeds by introducing PanelPOMP models, the mathematical framework we use to discuss mechanistic models for panel data; 
elsewhere, these models have also been called multi-SSMs \citep{wigren22}.
We then establish theoretical guarantees for the convergence of iterated filtering algorithms for these models (Theorem~\ref{theorem:pif}).
Our theoretical framework extends the analysis by \citet{chen25} to panel data while requiring weaker conditions than those by \citet{breto20}. 
We then introduce the marginalization step of MPIF that facilitates high dimensional inference by describing a concept we call marginalized Bayes maps.
The marginalization step introduces a nonlinearity in the Bayes map that existing theoretical approaches cannot immediately address.
We provide theoretical guarantees of the MPIF algorithm for some special cases, and demonstrate that the algorithm remains effective in a more general setting via a simulation study.
We then conduct a data analysis of a high-dimensional dataset of pre-vaccination measles case reports from 20 towns in the United Kingdom (UK).

\section{Iterated filtering for panel models}\label{sec:ppomp}

A POMP model comprises an unobservable Markov process $\{X(t), \, t \in \mathcal{T}\subset \mathbb{R}\}$ and an observable sequence $Y_1,\dots,Y_N$.
We suppose that $Y_n$ is a measurement of $X(t_n)$, with $t_1<\dots<t_N$, formalized by a requirement that $Y_n$ is conditionally independent of $\{X(s),Y_k:k\neq n, s\neq t_n \}$ given $X(t_n)$.
We assume that $X(t)$ takes values in $\mathcal{X} \subset \mathbb{R}^{d_x}$, and $Y_n$ takes values in $\mathcal{Y} \subset \mathbb{R}^{d_y}$.
Data collected at time $t_n$, denoted by $y^*_n$, are modeled as a realization of $Y_n$.

While $X(t)$ may be a continuous or discrete time process, the value of $X(t)$ at observation times is of particular interest, and so we write $X_n = X(t_n)$.
Initial values may be specified at a time $t_0<t_1$, and we set $X_0 = X(t_0)$.
We adopt the notation that for any integers $a$ and $b$, $\seq{a}{b}$ is the vector $(a, a+1, \ldots, b-1, b)$, and use the convention that $\seq{a}{b} = \emptyset$ if $b < a$. 
Similarly, we use the notation in subscripts to denote collections of random variables, such as $X_{0:N} = (X_0, \ldots, X_N)$, and use the same basic notation for $Y_{1:N}$, and $t_{0:N}$.  
We assume that the joint probability density $f_{X_{0:N}, Y_{1:N}}(x_{0:N}, y_{1:N}; \, \theta)$ exists, with respect to Lebesgue measure, or a counting measure if the set of possible values is discrete.

In a panel data analysis scenario, data are collected for each unit $\unit \in \seq{1}{\Unit}$, where $\Unit$ is the number of units in the panel.
Because the data generating process of each unit is assumed to be dynamically independent, we may model the panel data as a collection of related POMP models, which we call a PanelPOMP.
To distinguish between each unit POMP model, we denote the measurement and latent process for each unit with subscript $\unit \in \seq{1}{\Unit}$.
Specifically, we write $X_{\unit, n}, Y_{\unit, n}$ to denote the values of the latent process $\{X_{\unit}(t), \, t \in \mathcal{T}_\unit \subset \R\}$ and measurement sequence for unit $\unit$ at the measurement times $t_{\unit, 1} < t_{\unit, 2} < \ldots < t_{\unit, N_\unit}$.
The number of observations $N_\unit$, observations times $t_{\unit, n}$, and time domain $\mathcal{T}_u$ do not need to match across units, though this is often the case.
We use the shorthand $\bm{X}, \bm{Y}$ to refer to the collection of all latent and observable process at observation times for all units in the panel.
Due to the independence of units, the joint density of the PanelPOMP model can be written as
\begin{align}
f_{\bm{X}, \bm{Y}}(\bm{x}, \bm{y}; \, \theta) &= \prod_{u = 1}^U f_{X_{u, 0:N_u}, Y_{u,1:N_u}}(x_{u, 0:N_u}, y_{u, 1:N}; \, \theta) \label{eq:ppomp} \\
&\hspace{-2cm} = \prod_{u = 1}^Uf_{X_{u, 0}}(x_{u, 0};\, \theta)\prod_{n = 1}^{N_u}f_{Y_{u, n}|X_{u, n}}(y_{u, n} | x_{u, n} ;\, \theta)f_{X_{u, n} | X_{u, n-1}}(x_{u, n}|x_{u, n-1}; \, \theta), \label{eq:ppomp2}
\end{align}
where Eq.~\ref{eq:ppomp} arises from the assumption of dynamically independent units, and Eq.~\ref{eq:ppomp2} is a result of the Markov property of $X_{u}(t)$ and the conditional independence of $Y_{u, n}$.

Our goal is to make inferences about the parameter $\theta$ by maximizing the likelihood function
\begin{align}
L(\theta;\bm{y}^*) &= \int f_{\bm{X}, \bm{Y}}(\bm{x}, \bm{y}^*; \, \theta)\, d\bm{x}, \label{eq:likelihood}
\end{align}
where $\bm{y}^*$ denotes the collection observations from the entire panel.
Although the unit densities in Eq.~\ref{eq:ppomp} can be factored as a result of the dynamic independence between measurement units, a defining feature of a panel model is that the parameter vector $\theta$ remains relevant across all the dynamic systems.
Therefore, inference methodology used to obtain estimates of $\theta$ should use data from each unit in the PanelPOMP.
A special case of particular interest arises when $\theta = (\phi, \psi_{1:U})$, where $\phi \in \Theta_\phi \subset \mathbb{R}^{\nshared}$ is a vector of parameters that are shared across each unit, and $\psi_{u} \in \Theta_\psi \subset \mathbb{R}^{\nspecific}$ are parameters that are specific to unit $u$, formally,
$
f_{X_{u, 0:N_u}, Y_{u,1:N_u}}(x_{u, 0:N_u}, y_{u, 1:N}; \, \theta) = f_{X_{u, 0:N_u}, Y_{u,1:N_u}}(x_{u, 0:N_u}, y_{u, 1:N}; \, \phi, \psi_u).
$

Algorithm~\ref{alg:mpif} describes an iterated filtering algorithm for PanelPOMP models that obtains the MLE for both shared $(\phi)$ and unit specific $(\psi_u)$ parameters.
With $\mathrm{MARGINALIZE}=\mathrm{FALSE}$, this is the PIF algorithm of \citet{breto20}.
Our innovation occurs when $\mathrm{MARGINALIZE}=\mathrm{TRUE}$, and we call this the  marginalized panel iterated filter (MPIF). 
This small difference requires new approaches to theoretical analysis, and it is not immediately clear if an when this sequence of approximations may converge to the exact MLE.
However, we demonstrate that this modification has dramatic consequences for scalability.

Mathematically, MPIF adds an additional step to PIF by marginalizing out the unit-specific parameters that are irrelevant for the unit currently being filtered.
Because the parameter distributions are represented via Monte Carlo samples, marginalization is carried out by decoupling elements of the particles vectors representing parameters $\psi_{-u} = \{\psi_{k}\}_{k\neq u}$ from the elements of the vector representing $(\phi, \psi_u)$.
This decoupling can occur by not updating the particles representing $\psi_{-u}$ when using weights obtained from data in unit $u$.
In other words, instead of resampling all parameter particles in line~\ref{mpif:update} of Algorithm~\ref{alg:mpif}, we only resample particles for the shared parameters $\phi$ and the unit-specific parameters $\psi_u$ related to the data of the specific unit under consideration.

\begin{algorithm}[H]
   \caption[MPIF]{\textbf{Iterated Filtering for PanelPOMP models} \\
    {\bf Inputs}:\\\hspace{\textwidth}
    Simulator of initial density, $f_{X_{u, 0}}(x_{u, 0}; \, \theta)$ for $u$ in $\seq{1}{U}$.\\\hspace{\textwidth}
    Simulator of transition density, $f_{X_{u, n}|X_{u, n-1}}(x_{u, n}|x_{u, n-1}; \, \theta)$ for $u$ in $\seq{1}{U}$, $n$ in $\seq{1}{N_u}$.\\\hspace{\textwidth}
    Evaluator of measurement density, $f_{Y_{u, n}|X_{u, n}}(y_{u, n}|x_{u, n} ; \, \theta)$ for $u$ in $\seq{1}{U}$, $n$ in $\seq{1}{N_u}$.\\\hspace{\textwidth}
    Data $y_{u, n}^*$,for $u$ in $\seq{1}{U}$, $n$ in $\seq{1}{N_u}$.\\\hspace{\textwidth}
    Number of iterations, $M$.\\\hspace{\textwidth}
    Number of particles, $J$.\\\hspace{\textwidth}
    Starting parameter swarm, $\Theta_j^0 = \big(\Phi_j^0, \Psi_{1:U, j}^0\big)$ for $j \in \seq{1}{J}$, $u \in \seq{1}{U}$.\\\hspace{\textwidth}
    Simulator of perturbation densities, $h_{u, n}(\cdot | \varphi; \sigma)$ for $m \in \seq{1}{M}$, $u \in \seq{1}{U}$, $n \in 0:N_u$. \\\hspace{\textwidth}
    Perturbation Sequence $\sigma_{1:U, 1:M}$.\\\hspace{\textwidth}
    Logical variable determining marginalization, MARGINALIZE. \\\hspace{\textwidth}
    {\bf Output:} \\\hspace{\textwidth}
    Final parameter swarm, $\Theta_{j}^{m} = \big(\Phi_j^m, \Psi_{\seq{1}{U}, j}^m\big)$ for $j \in \seq{1}{J}$, $u \in \seq{1}{U}$.
    \label{alg:mpif}}
\For{$m \in 1:M$}{
  Set $\Theta_{0, j}^{F, m} = \Theta_{j}^{m-1} = (\Phi_j^{m-1}, \Psi_{1:U, j}^{m-1})$ for $j \in \seq{1}{J}$\;
    \For{$u \in \seq{1}{U}$}{
        Set $\Theta_{u, 0, j}^{F, m} = (\Phi_{u, 0, j}^{F, m}, \Psi_{1:U, 0, j}^{F, m}) \sim h_{u, 0}\big(\cdot | \Theta^{F, m}_{u-1, j}; \sigma_{u, m}\big)$ \label{line:startu}\; 
      Initialize $X_{u, 0, j}^{F, m} \sim f_{X_{u, 0}}(x_{u, 0}; \Phi_{u, 0, j}^{F, m}, \Psi_{u, 0, j}^{F, m})$ for $j \in \seq{1}{J}$\;
        \For{$n \in \seq{1}{N_u}$} {
          Set $\Theta_{u, n, j}^{P, m} = (\Phi_{u, n, j}^{P, m}, \Psi^{P, m}_{1:U, n, j}) \sim h_{u, n}\big(\cdot |\Theta_{u, n-1, j}^{F, m}; \sigma_{u, m}\big)$ for $j \in \seq{1}{J}$ \label{line:perturbations}\;
          $X_{u, n, j}^{P, m} \sim f_{X_{u, n}|X_{u, n-1}}\big(x_{u, n}|X_{u, n-1, j}^{F, m}; \Phi_{u, n, j}^{P, m}, \Psi_{u, n, j}^{P, m}\big)$ for $j \in \seq{1}{J}$ \label{line:Xpred}\;
          $w_{u, n, j}^m = f_{Y_{u, n}|X_{u, n}}\big(y_{u, n}^*|X_{u, n, j}^{P, m};\Phi_{u, n, j}^{P, m}, \Psi_{u, n, j}^{P, m}\big)$ for $j \in \seq{1}{J}$ \label{line:weights}\;
          Draw $k_{\seq{1}{j}}$ with $P(k_j = i) = w_{u, n, i}^m / \sum_{v = 1}^J w_{u, n, v}^m$ for $i, j \in \seq{1}{J}$\;
          Set $X_{u, n, j}^{F, m} = X_{u, n, k_j}^{P, m}$, and $\big(\Phi_{u, n, j}^{F, m}, \Psi_{u, n, j}^{F, m}\big) = \big(\Phi_{u, n, k_j}^{P, m}, \Psi_{u, n, k_j}^{P, m}\big)$ for $j \in \seq{1}{J}$\;
            \uIf{$\mathrm{MARGINALIZE}$}{$\Psi^{F,m}_{\tilde u,n,j} = \Psi^{P,m}_{\tilde u,n,j}$ for all $\tilde u \neq u$, $j=\seq{1}{J}$\label{mpif:update}} 
       \Else{$\Psi^{F,m}_{\tilde u,n,j} = \Psi^{P,m}_{\tilde u,n,k_j}$ for all $\tilde u \neq u$, $j=\seq{1}{J}$}
        }
      Set $\Theta_{u, j}^{F, m} = \big(\Phi_{u, N_u, j}^{F, m}, \Psi_{u, N_{1:U}, j}^{F, m}\big)$ for $j \in \seq{1}{J}$ \label{line:endu}\;
    }
  Set $\Theta_j^{(m)} = \Theta_{U, j}^{F, m}$ for $j \in \seq{1}{J}$\;
}
\end{algorithm}

The algorithmic complexity of both PIF and MPIF is $O(JMNU)$, where $M$ represents the number of iterations, $J$ is the number of particles, $U$ the number of units in the panel, and $N$ is the mean of $\{N_1, \ldots, N_U\}$.
Because MPIF does not require tracking the history of each particle, we achieve a minor reduction in computational overhead associated with the PIF algorithm.
However, as discussed in Section~\ref{sec:depletion}, the primary advantage of MPIF is that it typically exhibits lower Monte Carlo uncertainty (requiring smaller $J$) and generally converges in fewer iterations (requiring smaller $M$). 

Central to the success of iterated filtering algorithms is the fact that the repeated use of the posterior distribution from one Bayesian update as the prior distribution for the next iteration ultimately leads to a degenerate distribution centered at the MLE (see Section~\ref{sec:bayes}).
This idea is key to deriving the proof of Theorem~\ref{theorem:pif}, which extends existing theory \citep{chen25} for the convergence of the PIF algorithm. 

\begin{theorem}\label{theorem:pif}
  Consider a PanelPOMP model defined by Eq.~\ref{eq:ppomp}, and let $\Theta \subset \R^{\nshared + \Unit\nspecific}$ be a compact set that satisfies condition~\ref{assumption:regular} in Appendix~\ref{sec:assumptions}, and assume there exists a $\delta > 0$ such that $\{\theta \in \Theta: |\theta - \hat{\theta}|_2 < \delta\} \subseteq \Theta$, where $\hat{\theta} = \argmax_{\theta \in \Theta} L(\theta;\bm{y}^*)$.
    Denote the output of Algorithm~\ref{alg:mpif} without marginalization as $\Theta_{1:J}^{(M)}$, and assume that the model satisfies conditions~\ref{assumption:mle1}--\ref{assumption:mle3}, and the sequence of perturbation defined by $h_{u, n}(\cdot | \varphi, \sigma_{1:U, 1:M})$ satisfy conditions~\ref{assumption:kernel1}--\ref{assumption:kernel4}.
  Then there exists some positive sequences $\{C_M\}_{M \geq 1}$ and $\{\epsilon_M\}_{M \geq 1}$ where $\lim_{M \rightarrow \infty}\epsilon_M = 0$ such that for all $(J, M) \in \mathbb{N}^2$, 
  $$
  E\bigg[\Big|\frac{1}{J}\sum_{i=1}^J \Theta_j^{(M)} - \hat{\theta}\Big|_2\bigg] \leq \frac{C_M}{\sqrt{J}} + \epsilon_M
  $$
\end{theorem}
\begin{proof}[Proof outline]
  This theorem is an application of Theorem~4 of \citet{chen25} to PanelPOMP models, following the approach of \citet{breto20}, who showed that PanelPOMP models can be expressed as a lower-dimensional POMP model. 
  A full proof is provided in Appendix~\ref{sec:panelTheory}.
\end{proof}

Theorem~\ref{theorem:pif} formally provides guarantees for a variety of iterated filtering algorithms for PanelPOMP models, but as we demonstrate in Section~\ref{sec:sims}, the applicability of these algorithms often have scalability issues.
The marginalization step in Algorithm~\ref{alg:mpif} mitigates these issues but introduces a nonlinearity that invalidates the data cloning principle that is key to previous theoretical work on iterated filtering algorithms. 
In the following section, we briefly discuss data cloning and its relationship to the MPIF algorithm.

\section{Iterating marginalized Bayes maps}\label{sec:bayes}

If we denote $\pi_i(\theta)$ as the posterior distribution of the parameter vector $\theta$ after the $i$th Bayesian update, and $\bm{y}^*$ as the observed data, we can express an iterated Bayesian update as the following:
\begin{align*}
\pi_1(\theta) &\propto f(\bm{y}^* ;\, \theta)\,\pi_0(\theta), \\
\pi_2(\theta) &\propto f(\bm{y}^* ;\, \theta)\, \pi_1(\theta) \propto f^2(\bm{y}^* ;\, \theta)\, \pi_0(\theta),\\[-1em]
&\,\,\, \vdots \\[-1em]
\pi_m(\theta) &\propto f^m(\bm{y}^* ;\, \theta)\, \pi_0(\theta).
\end{align*}
In this representation, $f(\bm{y}^*; \,\theta)$ is the likelihood function (Eq.~\ref{eq:likelihood}), and $\pi_0(\theta)$ is the original prior distribution for $\theta$.
If we let $m\rightarrow \infty$, the effect of the initial prior distribution diminishes, and the $m$th posterior has all of its mass centered at the MLE.
This can be shown by taking the limit of $\pi_m(\theta) / \pi_m(\hat{\theta})$ as $m$ goes to infinity: if $\theta = \hat{\theta}$, then the limit is one, and zero otherwise \citep{lele07}.

The data cloning algorithm is useful for estimating the MLE in situations where the likelihood function is known or readily evaluated up to a constant of proportionality.
In this scenario, practitioners can leverage existing Bayesian software in order to obtain a maximum likelihood estimate (MLE) \citep{auger21,ponciano09}.
However, for nonlinear non-Gaussian state-space models, the likelihood function is generally inaccessible \citep{haggstrom25}.
In these cases, simulation-based inference techniques such as the particle filter \citep{arulampalam02} can reliably approximate the likelihood function.
An iterated Bayes procedure could in theory be used in conjunction with a particle filter to estimate the MLE.
However, a well-known issue with particle filters is the difficulty in accurately sampling the posterior distribution of fixed model parameters due to particle depletion \citep{Liu01}.
Consequently, a direct application of the iterated Bayes approach using particle filters for MLE estimation is impractical.

Iterated filtering algorithms overcome the issues associated with particle depletion by introducing a random walk for model parameters, thereby rescuing the degenerate particle representation of model parameters.
Early analysis of this procedure showed that if the random walk standard deviations are small, then this modification still leads to an approximation of the iterated Bayes algorithm outlined above, where the final particle mass is still centered at the MLE \citep{ionides15}.
The introduced perturbations of parameter values are necessary for inference, but they also introduce a loss of information \citep{Liu01}.
Therefore, in practice, the random perturbations are reduced as the number of iterations increases.
Recent theoretical analysis demonstrates that the algorithm still concentrates on the MLE when perturbations are reduced over time \citep{chen25}.

Because PanelPOMP models are a special case of POMP models, this same approach can theoretically be used to estimate the MLE.
However, the particle filter famously suffers from the curse-of-dimensionality.
That is, the approximation error of the particle filter grows exponentially with the number of units in a panel model \citep{bengtsson08,snyder08}.
The panel iterated filter (PIF) of \citet{breto20} partially addresses this issue by stacking the individual time series into a single long time series, and modifying the perturbation kernel for parameter particles to mitigate loss of information, as discussed in Section~\ref{sec:discussion}.
Importantly, the PIF algorithm still approximates the iterated Bayes map when the random walk standard deviations are small (Theorem~1 of \citet{breto20}), and can be shown to converge when the perturbations shrink over time (our Theorem~\ref{theorem:pif}).

The additional marginalization step in the MPIF algorithm introduces a nonlinear transformation at each iteration, changing the distribution that is approximated by the algorithm. 
As such, existing theoretic approaches for iterated filtering algorithms \citep{chen25} for this class of models are insufficient to demonstrate convergence of the algorithm. 
We explore why this is the case in the context of \emph{marginalized data cloning}. 
Let $\theta = (\phi, \psi_{1:U})$ denote the parameter vector for the panel model, where $\phi$ denotes the parameters that are shared by all $U$ units, and $\psi_u$ are the parameters that are only relevant to unit $u$.
Similarly, we write $y_u^* = y^*_{u, 1:N_u}$ to denote the time series data for unit $u$.
By stacking the times series into a single time series and iteratively filtering one at a time and ignoring parameter perturbations, the Bayes map that is approximated by a single sub-iteration of PIF can be written sequentially as:
\begin{align}
\pi_{m, u}(\theta) &\propto f_{u}(y^*_u;\, \theta)\, \pi_{m, u-1}(\theta) = f_{u}(y^*_u;\, \phi, \psi_u)\, \pi_{m, u-1}(\theta), \label{eq:PIFupdate}
\end{align}
where $\pi_{m, u}(\theta)$ is the parameter distribution at step $(m, u)$, and we adopt the convention that $\pi_{0, 0}(\theta) = \pi_0(\theta)$ is the initial prior density and $\pi_{m, U} = \pi_{m + 1, 0} = \pi_{m + 1}$.
This update is completed for each unit $u \in \seq{1}{U}$ which we call \emph{unit} iterations, then for iterated for $m \in \seq{1}{M}$, which we call \emph{complete} or \emph{full} iterations. 

An important observation regarding the representation in Eq.~\ref{eq:PIFupdate} is that each unit-specific likelihood function $f_{u}(y^*_u;\, \phi, \psi_u)$ contributes directly to the information about the shared parameter vector $\phi$ and its respective unit-specific parameter vector $\psi_u$, but not that of the unit-specific parameters $\psi_{-u}$.
Despite this, the Bayes update in the PIF algorithm necessitates updating the posterior distribution of all parameters at each unit-iteration.
Because the parameter distributions are represented via Monte Carlo samples (called particles), this implies that the PIF algorithm re-weights particles representing $\psi_{-u}$ based on a likelihood that does not contain direct information about the parameters.
For example, within a single $\nmif$-iteration, the particles representing the distribution of $\psi_U$ have been resampled $\sum_{\unit = 1}^{U - 1}N_u$ times before encountering the data $y^*_U$, the only subset of data containing direct information about $\psi_U$.
This process can lead to significant particle depletion (see Figure~\ref{fig:depletion}), particularly if $U$ or $N_u$ are large.

If the prior distribution used in Eq.~\ref{eq:PIFupdate} is independent across parameters, then the posterior distribution of the sub-vector $\psi_{-u}$ will be unchanged.
In this case, the particles representing the density for these parameters do not need to be resampled, which would avoid the issue of particle depletion.
The use of independent priors is common practice in Bayesian statistics, but each complete iteration of Eq.~\ref{eq:PIFupdate} introduces parameter dependence via the likelihood function.
This observation leads to the proposal of the marginalized PIF algorithm (MPIF), where the intermediate posterior distributions are made independent by marginalization before use as a prior distribution in the subsequent unit-iteration. 
A representation of a unit-iteration following this approach is given in Eqs.~\ref{eq:margBayes} and~\ref{eq:MPIFupdate}.
\begin{align}
\tilde{\pi}_{m, u}(\theta) &\propto f_{u}(y^*_u;\, \phi, \psi_u)\, \pi_{m, u-1}(\theta) \label{eq:margBayes}\\
\pi_{m, u}(\theta) &\propto \int \! \tilde{\pi}_{m, u}(\theta) \, d\phi \, d\psi_u \, \times \int \! \tilde{\pi}_{m, u}(\theta) \, d\psi_{-u} \label{eq:MPIFupdate}.
\end{align}

Here we have described iterated filtering algorithms while ignoring parameter perturbations and Monte Carlo evaluations of the likelihood function.
Each of these components play an important role in the practicality of iterated filtering, but existing theoretical justifications rely on the convergence of data cloning to the MLE, and show that the convergence still holds in spite of the additional complexities.
A natural question is whether iterating Eqs.~\ref{eq:margBayes}--\ref{eq:MPIFupdate} results in a probability distribution with all mass centered at the MLE(s), similar to the case without marginalization. 
The nonlinearization introduced by the marginalization, however, adds difficulty to the task of calculating, or bounding, the density.
In particular, previous approaches that rely on the linearization of unnormalized Bayes updates \citep[e.g.,][]{ionides15} are no longer applicable.
In the following subsection, we show that iterating Eqs.~\ref{eq:margBayes}--\ref{eq:MPIFupdate} does converge to the MLE when the likelihood is Gaussian. 

\subsection{Consistency for Gaussian models}

For Gaussian models, conditioning and marginalization can be carried out exactly.
The properties of this analytically tractable special case is relevant to the broader class of models that is well approximated by Gaussian models, for example, models satisfying the widely studied property of local asymptotic normality (LAN) \citep{lecam00}.
We show in Theorem~\ref{theorem:GG} that MPIF for a Gaussian model converges to the exact MLE as long as unit-specific parameters are not highly informative about shared parameters.

As before, we assume there are $U \geq 1$ units, and the likelihood of each unit is defined by $L_u(\theta;\, y^*_u) = f_u(y_u^*; \, \phi, \psi_u)$, and the likelihood of the entire model is $L(\theta;\, \bm{y}^*) = \prod_{u = 1}^U L_{u}(\theta; \, y_u^*)$.
In what follows, we assume $\phi \in \R$ and $\psi_u \in \R$ for all $u \in \seq{1}{U}$ in order to ease the notation and analysis. 

\begin{theorem}
  Let $f_u(y_u^*; \, \phi, \psi_u)$ be the density that corresponds to a Gaussian distribution with mean $(\phi^*, \psi_u^*)$ and precision $\Lambda_u^* \in \R^{2\times2}$. 
Assume that $\Lambda^*_u$ satisfies
\begin{align}
\big[\Lambda^*_{u}\big]_{1, 2}^2 <
  \frac{
    4 \big[\Lambda^*_{u}\big]_{1, 1}\, \big[\Lambda^*_{u}\big]^2_{2, 2}
    \, \sum_{k=1}^U \big[\Lambda^*_{k}\big]_{1, 1}
  }{
    \Big( \big[\Lambda^*_{u}\big]_{2, 2} + \sum_{k=1}^U\big[\Lambda^*_{k}\big]_{1, 1}\Big)^2
  }. \label{eq:GausAssumption}
\end{align}
  If the initial prior density $\pi_0(\theta) = \pi_{0, 0}(\theta)$ corresponds to a Gaussian distribution with mean $\mu_0$ and covariance $\Sigma_0$, then the density of the $m$th iteration of Eq.~\ref{eq:MPIFupdate} corresponds to a Gaussian distribution with mean $\mu_m \in \R^{U+1}$ and covariance $\Sigma_{m} \in \R^{U+1\times U+1}$ such that $\mu_m \rightarrow (\phi^*, \psi_1^*, \ldots, \psi_U^*)$ and $\|\Sigma_m\|_{2} \rightarrow 0$. That is, the algorithm converges in probability to the MLE. \label{theorem:GG}
\end{theorem}

The proof of Theorem~\ref{theorem:GG} is included in Appendix~\ref{appendix:Gaus}.
At each iteration of the algorithm, the marginalization step results in a loss of information about the likelihood surface. 
In the Gaussian setting, this equates to setting the covariance term between the shared and unit-specific parameters to be zero before performing a Bayes update. 
The assumption in Eq.~(\ref{eq:GausAssumption}) therefore helps mitigate this loss of information by controlling the size of the covariance in the likelihood surface. 
If the data are transformed to ensure that the likelihood covariance matrix has ones on the diagonal, that is, 
\begingroup
\renewcommand{\arraystretch}{0.7}
$
\Sigma^*_u = (\Lambda_u^*)^{(-1)} = \begin{pmatrix} 1 & \rho \\ \rho & 1\end{pmatrix}, 
$
\endgroup
then $\big[\Lambda^*_{u}\big]_{1, 1} = \big[\Lambda^*_{u}\big]_{2, 2} = 1 / (1-\rho^2)$ and $\big[\Lambda^*_{u}\big]_{1, 2} = \big[\Lambda^*_{u}\big]_{2, 1} = -\rho/(1-\rho^2)$.
In this case, the convergence condition in Eq.~(\ref{eq:GausAssumption}) becomes $\rho < 2 / \big(\sqrt{U}(1 + 1/U)\big)$. 

The proof of Theorem~\ref{theorem:GG} shows that the condition in Eq.~(\ref{eq:GausAssumption}) is sufficient for a convergence guarantee, but it may not be necessary. 
Furthermore, even some asymptotic bias may be tolerable compared to alternative algorithms that fail to scale. 
As demonstrated in Section~\ref{sec:depletion}, the particle depletion suffered by PIF can result in MPIF obtaining a better approximation of the unmarginalized map than the PIF algorithm.
In this case, MPIF is preferable to PIF even if the marginalized map results in a small bias.

Theorem~\ref{theorem:GG} provides convergence results for the algorithm in the absence of parameter perturbations. 
Using a similar setup, we can now consider the behavior of the algorithm with perturbations added to model parameters at each step. 
Let $f * g$ denote the convolution of probability densities $f$ and $g$. 
We assume that $h_{u, m}(\theta)$ is some perturbation density, and we modify the marginalized Bayes maps by adding this random noise at each unit-iteration.
\begin{align}
\tilde{\pi}'_{m, u}(\theta) &\propto f_{u}(y^*_u;\, \phi, \psi_u)\, \big(\pi'_{m, u-1} * h_{u, m}\big)(\theta) \label{eq:margBayesPerturb}\\
\pi'_{m, u}(\theta) &\propto \int \! \tilde{\pi}'_{m, u}(\theta) \, d\phi \, d\psi_u \, \times \int \! \tilde{\pi}'_{m, u}(\theta) \, d\psi_{-u} \label{eq:MPIFupdatePerturb},
\end{align}
Corollary~\ref{corollary:perturbed} shows that, under similar conditions as Theorem~\ref{theorem:GG}, marginalized data cloning with perturbations also converges to a point mass at the MLE if the likelihood is Gaussian.

\begin{corollary}\label{corollary:perturbed}
  Consider the setup of Theorem~\ref{theorem:GG}. If the parameter perturbations are Gaussian with covariance matrix $\sigma^2_m \Sigma_0$ for some initial covariance matrix $\Sigma_0 \in \R^{(U+1) \times (U+1)}$ and sequence $\sigma^2_m = o(1/m)$, then the $m$th iteration of Eqs.~\ref{eq:margBayesPerturb} and \ref{eq:MPIFupdatePerturb} corresponds to a Gaussian distribution with mean $\mu'_m \in \R^{U+1}$ and covariance $\Sigma'_m \in \R^{U+1\times U+1}$. If $\hat{\theta} = (\phi^*, \psi^*_1, \ldots, \psi^*_{U})$ denotes the MLE, then $|\mu'_m - \hat{\theta}|_2 \rightarrow 0$ and $\|\Sigma'_m\|_2 \rightarrow 0$.
\end{corollary}

The convergence of the Eqs.~\ref{eq:margBayesPerturb}--\ref{eq:MPIFupdatePerturb} can be partially explained using a common heuristic in Bayesian analysis: a more dispersed prior typically results in a posterior distribution that more closely resembles the likelihood function. 
In an iterated Gaussian setting, adding noise at each step results in intermediate prior distributions that have the same mean, but larger variance.
Therefore each iteration of Eq.~\ref{eq:margBayesPerturb} is expected to result in a mean closer to the MLE than the case without perturbations (Eq.~\ref{eq:margBayes}). 
Following this logic, if the perturbations are chosen to ensure that they eventually approach zero, then the convergence of the unperturbed marginalized data cloning algorithm heuristically implies the convergence of the perturbed version of the algorithm, as the perturbations to the prior densities at each step result in larger movements of the posterior density toward the MLE.
The proof of Corollary~\ref{corollary:perturbed} in Appendix~\ref{appendix:perturbed} demonstrates that this is true for the Gaussian model and the chosen perturbation schedule. 

In principle, the marginalization procedure can be applied at various steps in the data cloning process and one can obtain similar convergence results. 
In Figure~\ref{fig:DC}, we demonstrate the difference between data cloning and marginalized data cloning for a two parameter model with Gaussian likelihoods and priors with only a single unit but applying the marginalization for all parameters.
This useful visualization demonstrates how, even when the likelihoods can be computed exactly, the marginalization only makes small modifications to the intermediate distributions. 

\begin{figure}[ht]
\begin{knitrout}
\definecolor{shadecolor}{rgb}{0.969, 0.969, 0.969}\color{fgcolor}

{\centering \includegraphics[width=\maxwidth]{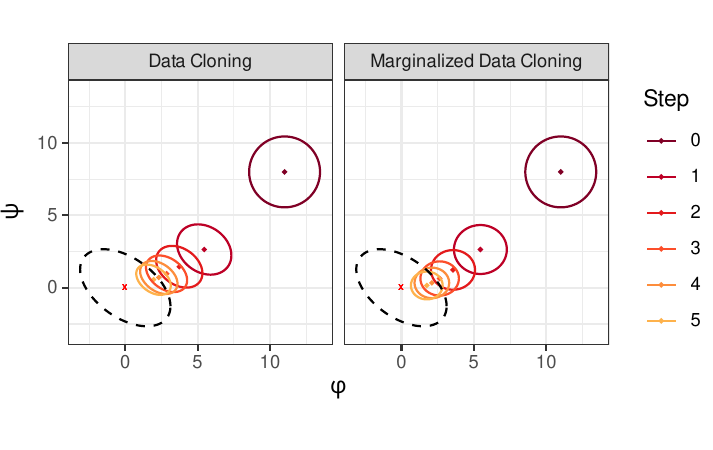} 

}

\end{knitrout}
\caption{\label{fig:DC}Data cloning and marginalized data cloning for two parameter model with Gaussian likelihoods and priors. The ellipses show the region of the parameter distribution that contains $95\%$ of the probability mass of the distribution. The black dashed line shows this region for the likelihood surface, and the red ``x" marks the MLE.
Theorem~\ref{theorem:GG} implies that the intermediate posterior densities will converge to a point mass at the MLE.
}
\end{figure}

\section{Simulation Studies}\label{sec:sims}

\subsection{Marginalization to reduce particle depletion}\label{sec:depletion}

The primary benefit and motivation of the marginalization step is improving the particle representations of the intermediate parameter distributions. 
In this sense, the marginalization step can be viewed as an attempt to take advantage of a bias-variance tradeoff. 
The marginalization procedure introduces a small amount of bias in the Bayesian posterior at each step in order to greatly reduce the variance of the particle representations of the distribution. 

We demonstrate this idea via a simple simulation study that explores the particle representation of parameter distributions with and without marginalization for only a single unit-iteration within Algorithm~\ref{alg:mpif}. 
For our model, we suppose $Y_{u, n}$ are independent and identically distributed (\iid) from a normal $\normal[\psi_u, 1]$ distribution, and do not specify a latent process model as it is irrelevant for this model.
We consider only $U = 2$ units, and $N_u = N = 100$ for all $u$.
For our prior distribution, we let $\Theta_{1:J}^{(0)} \overset{\text{\iid}}{\sim} \normal \big[ \mu_0, \, \Sigma_0 \big]$, and use $J = 1000$ particles to represent the joint parameter density.
This simple model and setup is selected so that the priors and likelihoods can be exactly calculated;
we can compare this to their particle representations using both versions of a single $u = 1$ iteration of Algorithm~\ref{alg:mpif} (Lines~\ref{line:startu}--\ref{line:endu}).

When iterating through unit $u = 1$, the Bayes map that is approximated by the un-marginalized filter requires an update to the particles that correspond to all model parameters for each time step $n \in \seq{1}{N_1}$. 
This reduces the number of unique particles that represent the intermediate posterior distributions for parameter $\Psi_2$ (Figure~\ref{fig:depletion}A). 
The number of unique particles representing $\Psi_1$ remains high as a result of the added parameter perturbations (line~\ref{line:perturbations}).
On the other hand, the MPIF algorithm does not require resampling the $\Psi_2$ particles during the $u = 1$ iteration, and thus maintains the same number of unique particles during this update (dashed horizontal line in Figure~\ref{fig:depletion}A). 

Figure~\ref{fig:depletion}B shows the filtered parameter particle swarm $\Theta_{1, 1:J}^{F, 1}$ after the single unit update under both versions of the algorithm compared to the Bayes posterior distribution that PIF approximates.
Although the marginalized version of the algorithm does not directly approximate this distribution, the particle swarm suffers less from particle depletion. 
This results in a better approximation of the intermediate posterior, despite introducing a small amount of bias.
Theorem~\ref{theorem:GG} provides sufficient conditions where the added bias at each step is negligible enough for the algorithm to converge to the MLE.

\begin{figure}[ht]
\begin{knitrout}
\definecolor{shadecolor}{rgb}{0.969, 0.969, 0.969}\color{fgcolor}

{\centering \includegraphics[width=\maxwidth]{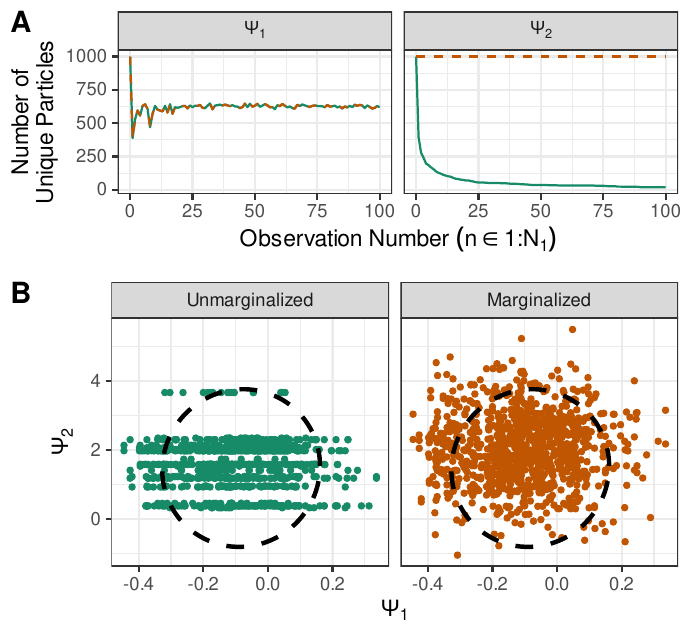} 

}

\end{knitrout}
\caption{\label{fig:depletion}Updating parameter distributions with a single $u = 1$ iteration of both versions of Algorithm~\ref{alg:mpif}. (A) The total number of unique particles representing each parameter. The dashed horizontal line shows that MPIF maintains the number of particles for $\Psi_2$ over time. (B) Parameter particle swarm of a single update with and without marginalization compared to the true posterior distribution.}
\end{figure}

\subsection{Stochastic Gompertz Population Model}\label{sec:gompertz}

We demonstrate the efficacy of the MPIF algorithm by estimating the model parameters of a high-dimensional, nonlinear PanelPOMP model by fitting data simulated from a collection of stochastic Gompertz population models.
This model is commonly used to describe population dynamics in Ecology and has been used as a benchmark for comparison between various algorithms in previous studies \citep{breto20}.
The model assumes a latent state vector $X_{u, n}$ for each unit $u \in 1:U$ and $n \in 0:N$.
For each unit $u$, the latent state has a one-step transition density that satisfies $X_{u,n+1} = K_u^{1-e^{r_u}} \, X_{u, n}^{e^{-r_u}}\, \epsilon_{u, n}$, where $K_u$ is the carrying capacity for the population in unit $u$, $r_u$ is a positive parameter that corresponds to the growth rate, and $\epsilon_{u, n}$ are $\iid$ log-normal random variables such that $\log \epsilon_{u, n} \overset{\iid}{\sim} \normal[0, \sigma_u^2]$.

Measurements of the population are obtained via the density $\log Y_{u, n} \overset{\iid}{\sim} \normal \left[\log X_{u, n}, \, \tau_u^2\right]$
where $\tau_u$ is a positive variance parameter.
This model is a convenient nonlinear non-Gaussian PanelPOMP model because a logarithmic transformation makes the model a linear Gaussian process, and as such the exact likelihood of the model can be calculated by the Kalman filter \citep{kalman60}.

For this simulation study, we generate data from several Gompertz population models, with equal number of observations $N$ in unit $u$, with values of $N \in \{20, 50, 100\}$ and values of $U$ ranging from $U=5$ to $U = 2500$.
To generate data from this model, we fix $K_u = 1$ and $X_{u, 0} = 1$ for all $u$ and treat these parameters as known constants.
We then set $\sigma_u^2 = 0.01$ and $r_u = 0.1$ for all values of $u$ to generate data, but treat these parameters as unknown shared parameters that need to be estimated from the data.
Finally, we set $\tau_u^2 = 0.01$ for all $u$ and treat these parameters as unknown unit-specific parameters.
These values were chosen to obtain comparable simulations and results as \citet{breto20}.

The models were fit using the MPIF algorithm with the number of iterations $M = 50$, and the number of particles $J = 1000$.
For this analysis, intermediate parameter estimates are obtained every five iterations of the MPIF algorithm, and the likelihood of the intermediate parameter values are obtained.
The goal of calculating the intermediate likelihood values is to demonstrate how many iterations are needed to obtain model convergence, and to compare the algorithms performance against that of the PIF algorithm for each time step.
Because the maximization procedure is inherently stochastic, it is recommended to try multiple starting parameter values.
Therefore for each model, 50 unique starting points are used; these starting points are randomly sampled from the hypercube with lower-bounds for each parameter determined by the generating parameter value divided by two, and the upper bound for the parameter defined as the generating parameter value multiplied by two.

The results of this simulation study with $N = 50$ are shown in Figure~\ref{fig:gomp}.
For every combination of $\{U, N\}$ considered, and for all numbers of MIF iterations, the maximum likelihood obtained using MPIF was higher than the maximum obtained using PIF. 
For large $U$, we also found that the worst performing Monte Carlo replicate of the MPIF algorithm often did better than the best performing replication of the PIF algorithm.
In the next section, we find similar results using a more complicated model to describe data that have previously been used as a means of comparing inference procedures.

\begin{figure}[!ht]
\begin{knitrout}
\definecolor{shadecolor}{rgb}{0.969, 0.969, 0.969}\color{fgcolor}

{\centering \includegraphics[width=\maxwidth]{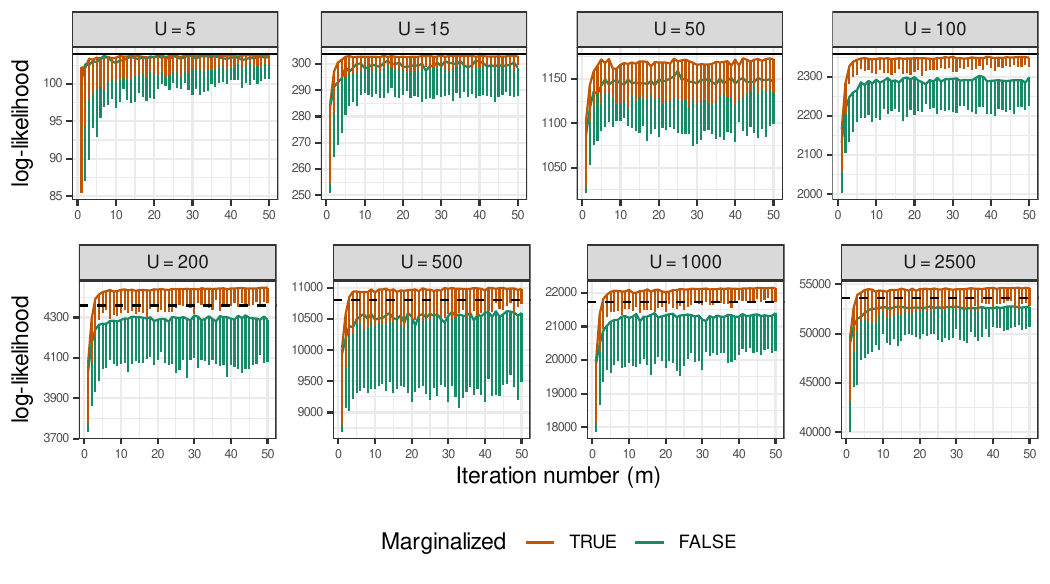} 

}

\end{knitrout}
\caption{\label{fig:gomp}Comparison of the MPIF and PIF algorithms for fitting the stochastic Gompertz population model. 
The solid horizontal line shows the true maximum likelihood, determined via the Kalman filter and a numeric optimizer, an intractable approach for high-dimensional parameter spaces $(U > 100)$; in these cases, the dashed line indicates the likelihood at the data-generating parameters. 
Each algorithm used 50 unique starting points.
Vertical bars span the tenth percentile to the maximum likelihood values}
\end{figure}

\section{Measles in the United Kingdom} 

We show how MPIF and PIF compare in an epidemiological model for weekly reported measles cases for 20 different UK cities from 1950 to 1964 \citep{korevaar20}.
Pre-vaccination UK measles data has been used extensively to motivate innovative methods for inference on stochastic processes since \citet{bartlett60}, yet, fitting nonlinear Markov process models simultaneously to multiple cities with shared and unit-specific parameters remains a challenge, leading practitioners to consider linearizations that have uncertain consequences \citep{korevaar20}.
We fit three different models to the data, all based on the susceptible-exposed-infectious-recovered (SEIR) model of \citet{he10}. 
The state process, $X_u(t)=\big(S^{(u)}_t,E^{(u)}_t, I^{(u)}_t, R^{(u)}_t\big)$, tracks the number of susceptible, exposed, infected and recovered individuals in each unit $u$. 
The total population size, $\pop^{(u)}(t)$, is treated as known, interpolated from census data using cubic splines, and we use this constraint to avoid explicit specification of $R^{(u)}_t$.
State transitions are generated using an Euler approximation to a continuous-time Markov chain, with time step $\eulerstep=1$ day, as follows:
\begin{align*}
A^{(u)}_{BS, t} &\sim \text{Pois}\big(\mu^{(u)}_{BS}(t)\, \eulerstep\big)
\\
(A^{(u)}_{SE, t},A^{(u)}_{SD, t}) &\sim \text{Eulermultinom}\Big(S^{(u)}_t, \bar{\mu}^{(u)}_{SE}(t)\big(\Gamma^{(u)}(t + \eulerstep) - \Gamma(t)\big)/\eulerstep, \mu^{(u)}_{SD}, \eulerstep \Big)
\\
(A^{(u)}_{EI, t},A^{(u)}_{ED, t}) &\sim \text{Eulermultinom}\big(E^{(u)}_t, \mu^{(u)}_{EI}(t), \mu^{(u)}_{ED},\eulerstep\big)
\\
(A^{(u)}_{IR, t},A^{(u)}_{ID, t}) &\sim \text{Eulermultinom}\big(I^{(u)}_t, \mu^{(u)}_{IR}(t), \mu^{(u)}_{ID}, \eulerstep\big),
\\
S^{(u)}_{t+\eulerstep} &= S^{(u)}_t + A^{(u)}_{BS, t} - A^{(u)}_{SE, t} - A^{(u)}_{SD, t}, 
\\
E^{(u)}_{t+\eulerstep} &= E^{(u)}_t + A^{(u)}_{SE, t} - A^{(u)}_{EI, t} - A^{(u)}_{ED, t}, 
\\
I^{(u)}_{t+\eulerstep} &= I^{(u)}_t + A^{(u)}_{EI, t} - A^{(u)}_{IR, t} - A^{(u)}_{ID, t}. 
\end{align*}
Here, $A^{(u)}_{BC,t}$ counts transitions from compartment $B$ into $C$ for unit $u$ between time $t$ and $t+\eulerstep$,  $\text{Pois}(\lambda)$ is a Poisson distribution with mean $\lambda$, $\Gamma^{(u)}(t)$ is a gamma process with mean $t$ and intensity $\sigma^{(u)}_{SE}$ \citep{breto09}, and $\text{Eulermultinom}(n, \mu_1,\mu_2,\eulerstep)$ is a multinomial distribution with $n$ independent trials and event probabilities given by $p_0 = \exp\big\{ -(\mu_1 + \mu_2)\eulerstep \big\}$ and $p_i = \frac{\mu_i}{\mu_1+\mu_2}\big(1-p_0\big)$.
The Eulermultinom outcome is the number of events of type $p_1$ and $p_2$, which correspond to transitions out of the source compartment.
The remaining events, of type $p_0$, correspond to individuals remaining in the source compartment.
The rate of arrivals into the susceptible class, $\mu^{(u)}_{BS}(t)$, is given by
\begin{align*}
	\mu^{(u)}_{BS}(t) = &\big( 1-c^{(u)} \big) \, b^{(u)}(t- \tau_d) +
	  c^{(u)} \, \delta\big( (t-\tau_c) \text{ mod } \tau_y\big)
	    \int^t_{t-\tau_y}b^{(u)}(t- \tau_d -s)\, ds,
\end{align*}
where $\delta$ is the Dirac delta function, $b^{(u)}(t)$ is the births per year at time $t$ interpolated from annual data using cubic splines, $\tau_d$ is the delay between when the births take place and when they actually contribute to the transition rate, $\tau_y = 1$ year, and $\tau_c$ is the school admission day, i.e., the 251st day of the year. 
We set $\tau_d = 4$ years to describe the age at which individuals typically enter a high-transmission environment.

The rate at which individuals enter state $E$ is $\mu^{(u)}_{SE}(t) = \bar{\mu}^{(u)}_{SE}(t) \, \frac{d\Gamma^{(u)}(t)}{dt}$, where $\bar{\mu}^{(u)}_{SE}(t)$ is the mean rate given by
$
  \bar{\mu}^{(u)}_{SE}(t) =
    \frac{
      \beta^{(u)}(t)\, \big(I^{(u)}_t+\iota^{(u)}\big)
    }{
      \pop^{(u)}(t)
    },
$
where $\pop^{(u)}(t)$ is the city population at time $t$ interpolated from annual data, and $\iota^{(u)}$ is the average number of infected individuals visiting the city at any time. 
The force of infection, $\beta^{(u)}(t)$, is given by
\begin{align*}
	&\beta^{(u)}(t) = \begin{cases}
		\beta^{(u)}_0 \big( 1+a^{(u)}(1-p)/p \big) & \text{during school term}\\
		\beta^{(u)}_0 \big( 1-a^{(u)} \big) & \text{during vacation}\\
	\end{cases}\\
  &\beta^{(u)}_0 = \mathcal{R}^{(u)}_0\big( 1-\exp\big\{ -(\mu^{(u)}_{IR}(t)+\mu^{(u)}_{ID})\Delta\big\}\big)/\Delta.
\end{align*}
where $p = 0.7589$ is the proportion of the year occupied by the school term.
The remaining transition rates are assumed to be constant: $\mu^{(u)}_{EI}(t) = \sigma^{(u)}$, $\mu^{(u)}_{IR}(t) = \gamma^{(u)}$, and $\mu^{(u)}_{SD} = \mu^{(u)}_{ED} = \mu^{(u)}_{ID} = 0.02 \, \text{yr}^{-1}$.
We follow \citet{he10} by using a discretized normal distribution for the number of cases reported: 
\begin{align*}
 P\big(Y^{(u)}_n = y^{(u)} \big| Z^{(u)}_n = z^{(u)} \big) =
   \Phi\big(y^{(u)} + 0.5, \rho^{(u)} z^{(u)}, \rho^{(u)}(1-\rho^{(u)})z^{(u)} +
     [\psi^{(u)} \, \rho^{(u)} \, z^{(u)}]^2 \big)\\
   \qquad - \Phi\big( y^{(u)} - 0.5, \rho^{(u)} z^{(u)}, \rho^{(u)}(1-\rho^{(u)})z^{(u)} +
     [\psi^{(u)}\, \rho^{(u)}\, z^{(u)}]^2 \big)
\end{align*}
where $\Phi(\cdot;\mu, \sigma^2)$ is the CDF for a $\normal[\mu,\sigma^2]$ random variable and $Z_n$ is the number of people transitioning from compartment $I$ to $R$ between the $(n-1)th$ and $n$th observation times. 
Lastly, we estimate the proportion of individuals in the first three states at time $t_0$, $S_0^{(u)}$, $E_0^{(u)}$, and $I_0^{(u)}$. 
Given that these proportions along with $R_0^{(u)}$ must add up to 1, there is no need to actually estimate $R_0^{(u)}$.
Consequently, the total number of parameters per unit is 12.

We use three different models based on the \cite{he10} model, with the key difference of investigating subsets of the parameters that might be well modeled as shared.
\begin{enumerate}
  \item The ``$c$-shared" model uses a shared parameter for $c$.
  \item The ``$\iota$-shared" model uses a log-log linear model between $\iota$ and the city population for year 1950, specifically $\log\big(\iota^{(u)}\big) = \iota_1 + \iota_2 \cdot \log\big(\pop^{(u)}(1950)\big)$.
  \item The ``7-shared" model uses the log-log linear model for $\iota$ and shared parameters for $c$, $\mathcal{R}_0$, $\gamma$, $\sigma$, $\sigma_{SE}$, and $a$.
\end{enumerate}

For each model, we run MPIF and PIF for 200 iterations, each search starting from one of 36 different parameter vectors randomly sampled from a hypercube where each dimension is slightly larger than the range spanned by the corresponding unit-specific MLE's in \cite{he10}. 
We perform these searches using 500, 5000, and 10000 particles to discern how the fits yielded by MPIF and PIF differ based on the selected particle count. 
For the present data set, 500 particles is too low for proper optimization, 5000 is adequate, and 10000 enables a thorough parameter search. 
The log-likelihood for each fit is evaluated using the average of replicated particle filter evaluations with 10000 particles at evenly-spaced iterations: 20, 56, 92, 128, 164, and 200. 
Figure~\ref{fig:measles1} summarizes the output of this Monte Carlo optimization search. 
Similar to the Gompertz population model, the MPIF algorithm consistently yields parameter estimates corresponding to higher likelihoods than the PIF algorithm.

In practice, because iterated filtering algorithms are stochastic optimization algorithms, many Monte Carlo replicates are conducted from various initialization points, and final parameter estimates correspond to the search with the highest likelihood. 
Because of this, we are primarily interested in how in the maximum estimate from each algorithm compares. 
In Figure~\ref{fig:measles1}, we note that MPIF consistently yields larger sample maximums of the log-likelihood evaluations irrespective of the model or particle count we use. 
By 200 iterations, we see that maximum log-likelihoods obtained using MPIF are 70 to 105 units higher than PIF for the $\iota$-shared and $c$-shared models, or about 4 to 5 units per city. 
MPIF has about half the advantage for the 7-shared model by 200 iterations. 
Close inspection of the results show that for nearly all combinations of number of particles, iterations, and model specification, MPIF outperforms the PIF algorithm.
In cases where this is not true, the distribution for PIF-generated log-likelihoods have especially long right-tails, suggesting that the observed advantage for PIF in these scenarios is a result of large variance working in its favor.
Increasing the number of particles and iterations reduces the variance among Monte Carlo replications, and when the number of particles ($J$) is largest, the advantage of MPIF over PIF becomes more evident for all models variations.

The boxplots in Figure~\ref{fig:measles1} demonstrate that the standard deviation of the log-likelihood across Monte Carlo replicates is consistently lower for MPIF. 
In addition to suggesting that even poor performing Monte Carlo replicates of the MPIF algorithm have better outcomes than the best replicates for the PIF algorithm, low standard deviation across initializations is useful in practice as the convergence of iterated filtering algorithms is often judged by whether or not distinct Monte Carlo replicates finish at the same maximum, within suitable Monte Carlo error. 
Reducing the variance between Monte Carlo replicates also results in tighter confidence intervals when computing Monte Carlo adjusted profile confidence intervals \citep{ionides17}.

\begin{figure}[ht]
\begin{knitrout}
\definecolor{shadecolor}{rgb}{0.969, 0.969, 0.969}\color{fgcolor}

{\centering \includegraphics[width=\maxwidth]{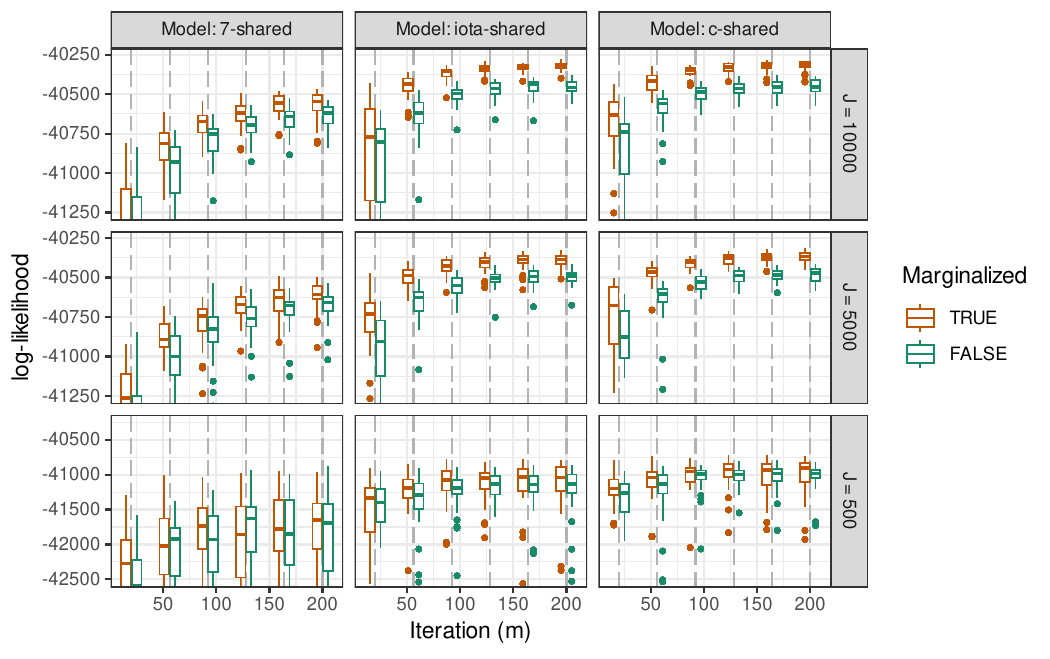} 

}

\end{knitrout}
\caption{\label{fig:measles1} Log-likelihoods yielded by PIF and MPIF for the mechanistic measles model to the UK data. 
Rows correspond to a different number of particles $J$ used in Algorithm~\ref{alg:mpif}. 
The log-likelihood is evaluated at iterations 20, 56, 92, 128, 164, and 200.}
\end{figure}

\section{Discussion}\label{sec:discussion}

The issue of particle depletion when estimating the posterior distribution of parameter values using a particle filter has previously been noted \citep[e.g.,][Section~1.2]{chen25}. 
A key innovation of iterated filtering algorithms is that adding parameter perturbations helps revive the particle representations from a depleted state by adding artificial noise to the parameter particles. 
These perturbations are precisely why the number of unique particles representing the distribution of $\Psi_1$ in Panel~A of Figure~\ref{fig:depletion} does not degenerate to zero over time.
Given this observation, an alternative approach to solving the particle depletion issue that arises in higher dimensions would be to perturb all model parameters at each time-step. 
This approach is supported by Theorem~\ref{theorem:pif}, and avoids the analytic challenges associated with adding the marginalization step.
However, perturbing all model parameters at each step is equivalent to applying a vanilla iterated filtering algorithm to a PanelPOMP model, which generally performs worse on panel models than the PIF variation \citep{breto20}. 

\citet{Liu01} note that the artificial noise introduced from parameter perturbations results in loss of information.
This loss of information motivates the common practice in applications of iterated filtering of only perturbing model parameters when the data used to calculate the particle weights has direct relation with the parameters that are being perturbed.
For instance, it is common practice to only perturb parameters that are unique to the initialization density $f_{X_{u, 0}}(x_{u, 0};\, \theta)$ at the first available observations, as these observations have the strongest signal for the initialization parameters. 
If initial parameters are perturbed at all observation times, then the signal from these initial observations gets lost over time. 
This same principle serves as a motivator for the PIF algorithm and describes why PIF is more successful than vanilla iterated filtering algorithms applied to PanelPOMPs: perturbing unit-specific parameters while considering data from other units results in a significant loss of information.

The challenge of iterating filtering for PanelPOMP models can be described as a tradeoff between particle depletion and a loss of signal due to parameter perturbations. 
The MPIF algorithm introduced here is designed to address both of these challenges simultaneously by avoiding perturbations when the signal is weak, but not resampling unit-specific parameters using weights calculated with data from other units. 
The cost of this modification is a small amount of bias for the particle representation of the posterior distribution at each step. 
Theorem~\ref{theorem:GG} formally demonstrates that the affect of the bias can be completely negated if the log-likelihood is quadratic, which is the limiting behavior of all likelihood surfaces. 

Finally, Theorem~\ref{theorem:pif} does provide some stronger guarantees for MPIF under additional constraints on the PanelPOMP model that we have not yet mentioned. 
In a model with only unit-specific parameters ($\Phi = \emptyset$), for instance, MPIF is equivalent to conducting IF2 independently on each unit, and therefore the convergence results of Theorem~\ref{theorem:pif} apply. 
Similarly, if the model only contains shared parameters ($\Psi_{1:U} = \emptyset$), then MPIF is equivalent to PIF, and once again stronger theoretical guarantees are available. 
This equivalency also provides some intuition as to when MPIF will outperform its alternatives. 
In a model consisting primarily of shared parameters, MPIF behaves very similarly to PIF and adds a smaller advantage relative to the case when there are more unit-specific parameters. 
We have seen this pattern in our results, as the gain in log-likelihood obtained via MPIF was smallest for the measles model with only one unit specific parameter.

\vspace{-3mm}

\bibliographystyle{jasa3}

\bibliography{references}

\newpage

\begin{appendices}

\section{Assumptions for Theorem~\ref{alg:mpif}}\label{sec:assumptions}

For all constants $(\epsilon, r) \in (0, \infty) \times \R^d$, we define $B^d_\epsilon(r) = \{r' \in \R^d: \vert r - r'\vert_2 < \epsilon\}$. 
Let $\BorelT$ be the Borel $\sigma$-algebra on the set of real numbers $\R^{\nshared+\Unit\nspecific}$.
We assume that $\Theta \in \BorelT$ is a compact set that satisfies \ref{assumption:regular}. 
Informally, this ensures that the corners of the set are not too sharp, and directly follows definition of a regular compact set from \citet{chen25}.
\begin{enumerate}[label=(A\arabic*), ref=(A\arabic*)]
  \item \label{assumption:regular} There exists a continuous function $\regFun\colon [0, \infty) \to [0, \infty)$ such that $\lim_{x\downarrow 0} \regFun(x) = 0$, and for all $(\epsilon, x) \in (0,\infty) \times \Theta$, there exists $x' \in \Theta$ such that
$$
B^{\nshared+\Unit\nspecific}_{\regFun(\epsilon)}(x') \subseteq B^{\nshared+\Unit\nspecific}_{\epsilon}(x) \cap \Theta.
$$
\end{enumerate}
We make the following assumptions on the probability densities that are used to define a PanelPOMP model described in Section~\ref{sec:ppomp}.
\begin{enumerate}[label=(B\arabic*), ref=(B\arabic*)]
\item There exists a compact set $E \subset X$ such that 
$$
\inf_{(\theta, x) \in \Theta \times X} \int_E f_{X_n|X_{n-1}}(x_n | x_{n-1}; \theta) dx_n > 0,
$$
for all $n \in \{1, 2, \ldots, N\}$. \label{assumption:mle1}
\item $L(\theta; \bm{y}^*) > 0$ for all $\theta \in \Theta$ and $\sup_{(\theta, x_{u, n})\in (\Theta, \mathcal{X})}f_{Y_{u, n}|X_{u, n}}(y^*_{u, n}|x_{u,n}; \, \theta) < \infty$ for all $u \in \seq{1}{U}$ and $n \in \seq{1}{N_u}$. \label{assumption:mle2}
    \item The transition and measurement densities are sufficiently smooth functions of $\theta$, in the sense that for any $\theta, \theta' \in \Theta$, there exists a a continuous and strictly increasing function $g: [0, \infty) \rightarrow [0, \infty)$ and sequence of measurable functions $\varphi_{u, 0:N_u}: \mathcal{X}^2 \rightarrow \mathbb{R}$, that satisfy:
    \begin{align*}
    \begin{split}
    |&\log\big(f_{Y_{u, n}|X_{u, n}}(y^*_{u, n}|x_{u, n};\, \theta)f_{X_{u, n}|X_{u, n-1}}(x_{u, n}|x_{u, n-1};\, \theta)\big) \\
    & - \log\big(f_{Y_{u, n}|X_{u, n}}(y^*_{u, n}|x_{u, n};\, \theta')f_{X_{u, n}|X_{u, n-1}}(x_{u, n}|x_{u, n-1};\, \theta')\big)| \leq g\big(\|\theta - \theta'\|\big)\varphi_{u, n}(x_{u, n-1}, x_{u, n}),
    \end{split}
    \end{align*}
    With the constraint that for all $u\in \seq{1}{U}$, there exists a $\delta_u \in (0, \infty)$ such that
    \begin{align*}
    &  \int \exp\Big\{\delta_u \sum_{n = 0}^{N_u}\varphi_{u, n}(x_{u, n-1}, x_{u, n})\Big\} f_{X_{u, 0}}(x_{u, 0};\, \theta) \times
    \\
    & \hspace{20mm} \prod_{n = 1}^{N_u}f_{Y_{u, n}|X_{u, n}}(y^*_{u, n}|x_{u, n};\, \theta)f_{X_{u, n}|X_{u, n-1}}(x_{u, n}|x_{u, n-1};\, \theta)\, dx_{u, 0:N_u} < \infty,
    \end{align*}
    using the convention that if $n=0$, then
    \begin{eqnarray*}
      \varphi_{u, n}(x_{u, n-1}, x_{u, n}) &=& \varphi_{u, 0}(x_{u, 0}, x_{u, 0})
      \\
      f_{Y_{u, n}|X_{u, n}}(\cdot | x_{u, n}; \, \theta) &=& 1
      \\
      f_{X_{u, n} | X_{u, n-1}}(x_{u, n}|x_{u, n-1};\, \theta) &=& f_{X_{u, 0}}(x_{u, 0}; \, \theta)
    \end{eqnarray*}\label{assumption:mle3}
\end{enumerate}
\vspace{-13mm}
Finally, the following assumptions are made about the random perturbations in lines~\ref{line:startu} and \ref{line:perturbations} of Algorithm~\ref{alg:mpif}.
Let $\mu_0$ denote the probability measure on $\big(\R^{\nshared + \Unit\nspecific}, \BorelT\big)$ that defines the distribution of the initial particle swarm, i.e., $\Theta_{1:J}^0 \overset{\iid}{\sim} \mu_0$. 
If $\{\mu_n\}_{n \geq 1}$ is a sequence of random probability measures \citep{crauel02} on $\big(\R^{\nshared + \Unit\nspecific}, \BorelT\big)$ with $\mathcal{F}_n$ denoting the corresponding filtration, then we denote $K_{\mu_n}$ to be the Markov kernel such that \linebreak $\theta \sim K_{\mu_n}(\theta', d\theta) \iff \theta \overset{dist}{=}\theta' + \addRV,$ where $\addRV|\mathcal{F}_n \sim \mu_n$.

Let $S_{\tilde{u}} = \sum_{k = 1}^{\tilde{u}} (N_{k} + 1)$. 
We define the Markov kernel as a sequence in $\nclone \in \mathbb{N}$, where $\nclone$ defines the values $(m, u, n)$ via the equation $\nclone = (\nmif-1)S_\Unit + S_{\unit-1}+n+1$, such that $K_{\nclone}(\theta_{\nclone-1}, d\theta_{\nclone}) = h_{u, n}(\theta_{\nclone}|\theta_{\nclone-1};\sigma_{u, \nmif})d\theta_{\nclone}$.
Let $\{U_{\nclone}\}_{\nclone \geq 1}$ be a sequence of $\Theta$-valued random variables such that for all $\nclone \geq 1$ and sets $A_1, \ldots, A_{\nclone} \in \mathcal{B}(\Theta)$, where $\mathcal{B}(\Theta)$ is the Borel $\sigma$-algebra on $\Theta$. 
    $$
    \prob\big(\addRV_k \in A_k, k\in \{1, \ldots, \nclone\} | \mathcal{F}_{\nclone}\big) = \prod_{k = 1}^{\nclone} \mu_k(A_k). 
    $$
Then we assume the sequence of probability measures $\{\mu_{\nclone}\}_{\nclone \geq 0}$ satisfy:
    \begin{enumerate}[label=(C\arabic*), ref=(C\arabic*)]
      \item $\mu_0\big(B_{\epsilon}(\theta)\big) > 0$ for all $\theta \in \Theta$ and $\epsilon \in (0, \infty)$.\label{assumption:kernel1}
      \item There exists a family of $(0, 1]$-valued random variables $(\Gamma^\mu_\delta)_{\delta \in (0, \infty)}$ such that, for all $\delta \in (0, \infty)$, we have $\prob\big(\inf_{\nclone \geq 1}\mu_{\nclone}\big(B_\delta(0)\big) \geq \Gamma_{\delta}^\mu\big) = 1$.\label{assumption:kernel2}
      \item $\prob\big(\inf_{\nclone \geq 1}\inf_{\theta' \in \Theta}\int_\Theta K_{\mu_{\nclone}}(\theta', d\theta) \geq \Gamma^\mu\big) = 1$ for some $(0, 1]$-valued random variable $\Gamma^\mu$.\label{assumption:kernel3}
      \item \label{assumption:kernel4} There exists a sequence of natural numbers $\{k_{\nclone}\}_{\nclone \geq 1}$ and, for all $l \in \mathbb{N}_0$, a sequence of $(0, \infty]$ valued functions $\{g_{\nclone, l}\}_{\nclone \geq 1}$ defined on $(0, \infty)$ such that \begin{enumerate} 
  \item $\lim_{\nclone \rightarrow \infty} k_{\nclone} / {\nclone} = \lim_{\nclone \rightarrow \infty} 1/k_{\nclone} = 0$, and for all $\epsilon \in (0, \infty)$, $\lim_{\nclone \rightarrow \infty}g_{\nclone, l}(\epsilon) = 0$. 
  \item For all $\nclone \geq 1$ such that $\nclone > 2k_{\nclone}$, $k^*_{\nclone} \in \{k_{\nclone}, \ldots, 2k_{\nclone}\}$, and for all $\epsilon \in (0, \infty)$,
  $$
  \frac{1}{\nclone - k^*_{\nclone}}\log \prob \bigg(\exists s \in \{k^*_{\nclone} + 1, \ldots, \nclone\}: \sum_{i = k^*_{\nclone} + 1}^s \addRV_i \notin B_{\epsilon}(0) | \mathcal{F}_{\nclone}\bigg) \leq -\frac{1}{g_{\nclone, l}(\epsilon)}.
  $$
  \item For all $\nclone > l$ and $\epsilon \in (0, \infty)$, we have 
  $$
  \frac{1}{\nclone-l}\log \prob \bigg(\sum_{i = l + 1}^{s} \addRV_i \in B_{\epsilon}(0), \forall s \in \{l + 1, \ldots, \nclone\}|\mathcal{F}_{\nclone}\bigg) \geq -g_{\nclone, l}(\epsilon).
  $$
  \end{enumerate}
\end{enumerate}

  \section{Proof and discussion of Theorem~\ref{theorem:pif}}\label{sec:panelTheory}

  Theorem~\ref{theorem:pif} is a straightforward extension of Theorem~4 of \citet{chen25} to PanelPOMP models. 
  Our approach is to express an arbitrary PanelPOMP model as a POMP model using a long format, where the latent and observed processes are stacked one unit after another to describe a single long time series that arises by stacking unit time series data one after another.
  Then, Algorithm~\ref{alg:mpif} is equivalent to applying the theory developed in the appendix of \citet{chen25} to this long POMP model.
  For instance, if the perturbation schedule in Algorithm~\ref{alg:mpif} is chosen to to follow the dynamic approach introduced by the authors, then Theorem~1 is just an application of Theorem~4 of \citet{chen25} to a panel version of the model. 
  
  The stacking argument to prove Theorem~\ref{theorem:pif} follows the approach of \citet{breto20}, who previously introduced the PIF algorithm and extended an existing theory for low-dimensional POMP models \citep{ionides15} to derive theoretical properties of PIF. 
  However, the recent work of \citet{chen25} provides convergence results for iterated filtering algorithms under weaker assumptions than \citet{ionides15}.
  Most notably, \citet{chen25} prove convergence of iterated filtering algorithms as the random walk standard deviation for parameter perturbations decreases over time rather than being fixed at a small constant.
  Thus, this approach allows us to provide stronger theoretical results for panel iterated filtering than obtained by \citet{breto20}.
  
  \begin{proof} 
  Recall the basic definition of the joint density of a PanelPOMP model: 
  \begin{align}
&f_{X_{0:N}, Y_{1:N}}(x_{0:N}, y_{1:N}; \, \theta) =
\nonumber
\\
& \hspace{20mm} \prod_{u = 1}^Uf_{X_{u, 0}}(x_{u, 0};\, \theta) 
\prod_{n = 1}^{N_u}f_{Y_{u, n}|X_{u, n}}(y_{u, n} | x_{u, n} ;\, \theta)f_{X_{u, n} | X_{u, n-1}}(x_{u, n}|x_{u, n-1}; \, \theta). \label{eq:ppompSI}
\end{align}
  We would like to pivot the unit specific processes into a single long process, avoiding the need to loop over the unit $u$.
  We define $S_{\tilde{u}} = \sum_{k = 1}^{\tilde{u}}(N_k + 1)$ to be the sum total number of time-steps (observations + initialization) for units $1$ up to unit $\tilde{u}$, defining $S_0 = 0$.
  We use the sequence $\nstack \in \mathbb{N}$ to map time points and states indexed with $(u, n) \in \seq{1}{U}\times\seq{0}{N_u}$ to a new model with only a single index $\nstack \in \seq{1}{S_U}$, defined by the equation $\nstack = S_{\unit-1} + n + 1$
  
  Now let $T_{\tilde{u}} = \sum_{k = 1}^{\tilde{u}} (t_{k, N_k} - t_{k, 0})$ for $\tilde{u} \in \seq{1}{U}$, with $T_0 = 0$.
  We now define new latent processes $\tilde{X}$ and $\tilde{Y}$. As the original model allows for continuous time latent process, we write 
  First,
  \begin{align}
    \tilde{X}(t) &= X_u(t_{u, 0} + (t - T_u)), \, \text{for} \,\, T_{u-1} \leq t \leq T_u.
  \end{align}
  Letting $\tau_{\nstack} = t_{u, n}$, the latent and observable processes are equivalent to
  \begin{align*}
    \tilde{X}_{\nstack} = \tilde{X}(\tau_{\nstack}) = X_{u, n}, \quad \text{and} \quad \tilde{Y}_{\nstack} = Y_{u, n}. 
  \end{align*}
  With this new definition of $\tilde{X}$ and $\tilde{Y}$, we have a new POMP model which has joint density
  \begin{align}
&f_{\tilde{X}_{1:S_U}, \tilde{Y}_{1:S_U}}(x_{1:S_U}, y_{1:S_U}; \, \theta) =  \nonumber
\\
& \hspace{20mm}
f_{\tilde{X}_0}(x_0;\, \theta)\prod_{\nstack = 1}^{S_{U}}f_{\tilde{Y}_{\nstack}|\tilde{X}_\nstack}(y_{\nstack} | x_{\nstack}; \, \theta)f_{\tilde{X}_{\nstack}|\tilde{X}_{\nstack - 1}}(x_{\nstack} | x_{\nstack - 1}; \, \theta), \label{eq:lowPOMP}
  \end{align}
  Using the convention that $n' \mapsto (u,0)$ for any $u$, then $f_{\tilde{Y}_{\nstack}|\tilde{X}_\nstack}(y_{\nstack} | x_{\nstack}; \, \theta) = 1$ and \linebreak $f_{\tilde{X}_{\nstack}|\tilde{X}_{\nstack - 1}}(x_{\nstack} | x_{\nstack - 1}; \, \theta) = f_{X_{u, 0}}(x_{u, 0}; \, \theta)$.

  As defined, Eqs.~\ref{eq:lowPOMP} and \ref{eq:ppompSI} describe the same model, but \ref{eq:lowPOMP} has been written to match the SSM of \citet{chen25}, after adjusting for the choice of initializing at $\tilde{X}_1$ rather than $\tilde{X}_0$. 
  We refer to Eq~\ref{eq:lowPOMP} as the \emph{long} format, which is indicative that the model has been expressed as a low-dimensional POMP model with observation times ranging from $1$ to $S_U$, rather than a collection (or product) of POMP models, each with observation times $N_u+1$ for $u \in \seq{1}{U}$. 
  
  This representation allows us to naturally extend the theoretical framework of \citet{chen25} to PanelPOMP models.
  Specifically, Assumption~\ref{assumption:regular} imply that the set $\Theta$ is a regular compact set \citep[See Definition~1 of][]{chen25}; Assumptions~\ref{assumption:mle2}--\ref{assumption:mle3} ensure that the density in Eq.~\ref{eq:lowPOMP} corresponds to a state-space model that satisfies the MLE conditions of \citet{chen25}, and Assumption~\ref{assumption:mle1} is used as a uniformity condition across particle representations. 
  Finally, Assumptions~\ref{assumption:kernel1}--\ref{assumption:kernel4} are applied to the cloned version of Model~\ref{eq:lowPOMP}. 
  The process is cloned such that we have a new SSM $\dbtilde{Y}_{\nclone} = Y_{u, n}$, $\dbtilde{X}_{\nclone} = X_{u, n}$, using the mapping $\nclone = (\nmif-1)S_\Unit + S_{\unit-1}+n+1$.
  Assumptions~\ref{assumption:kernel1}--\ref{assumption:kernel4} therefore imply assumptions C1-C3 and C4' of \citet{chen25} for the perturbation kernels of the self organized SSM defined via $\big(\dbtilde{Y}_{\nclone}, \dbtilde{X}_{\nclone}\big)$.
  An alternative version of Assumption~\ref{assumption:kernel4} is also stated in Appendix~A.5 of \citet{chen25}, but here we only present one version for brevity.
  Together, the conditions stated in Appendix~\ref{sec:assumptions} allow for a direct application of of the theory developed by \citet[][Appendix~A]{chen25}, with precise details of applying this theorem to the cloned model given in \citet[Supplement~S7,][]{chen25}.
    \end{proof}

  \noindent Formally, Theorem~\ref{theorem:pif} provides guarantees for several variants of iterated filtering algorithms applied to panel models, where each variant is a change to the perturbation kernel that satisfy the conditions \ref{assumption:kernel1}--\ref{assumption:kernel4}.
  For example, Theorem~4 of \citet{chen25} is stated for specific perturbation kernels, such as those based on a normal distribution---a choice that has been used by default in most iterated filtering applications \citep[e.g.,][]{ionides15,fox22,subramanian21,wheeler24}.
  This theory also applies to the variant of iterated filtering proposed by \citet{chen25} where perturbations are only applied when needed in a dynamic fashion.
  
  Conditions \ref{assumption:kernel1}--\ref{assumption:kernel4} are difficult to verify in practice, and not all variants of the algorithm that satisfy this condition are useful for panel models. 
    For instance, multivariate normal perturbations applied at each time point is equivalent to applying IF2 to panel models; this approach generally leads to worse results than the PIF algorithm applied to the same model \citep{breto20}.
    As pointed out in Section~\ref{sec:discussion}, the effect of the perturbation kernel is twofold. 
    First, it helps revive particle representations of the intermediate parameter distributions by adding random noise. 
    Second, the random noise results in a loss of information by masking the signal from the observed data. 
    These competing interests are not addressed in theorems involving iterated filtering algorithms.
    In this case, heuristics are useful for determining suitable perturbation densities.
    
    The modification of \citet{chen25} that applies perturbations only when needed is an effective approach to address the tradeoff between these competing interests.
    For high-dimensional panel models, however, this modification still leads to a large number of updates of the particle representation of the parameter vector $\psi_{-u}$ using data from unit $u$, which is the primary reason that algorithms like PIF struggle in higher dimensional settings.
    This observation leads to the proposal of the MPIF algorithm, which avoids resampling parameters if little information via the likelihood function.
 Combining the dynamic perturbations strategies of \citet{chen25} with a marginalization step may also result in an improved algorithm in some cases.
 However, we find that the default multivariate normal perturbations with singular covariance matrix to avoid perturbing all parameters at each time step to be sufficient for parameter estimation in practice.
  
One reason that Theorem~\ref{theorem:pif} cannot be used directly to infer the practicality of iterated filtering algorithms for panel data is because the behavior of $C_M$ as $M \rightarrow \infty$ is unknown.
Previous works on high-dimensional particle filtering---without adding perturbations or performing data cloning---suggest that the sequence $C_M$ scales exponentially with the number of units \citep{snyder08,bengtsson08}.
While this problem has partially been avoided by writing the PanelPOMP in a long format, which reduces the size of both the latent and observed spaces at each time point, the parameter space for $\Theta$ remains large.
Specifically, the total dimension of $\Theta$ is $\nshared + U\nspecific$, where $\nshared$ and $\nspecific$ are the number of shared and unit-specific parameters, respectively.

Because particle filters are known to perform poorly in high-dimensions, alternative filtering algorithms should be used.
One such example is called the block particle filter \citep{rebeschini15}; this algorithm breaks the state-space into separate units called \emph{blocks}, and performs the update step in the filtering equation independently on each block.
Block particle filters have been found to be effective in high-dimensional settings where the units can be treated as approximately independent. 
These results also provides an alternative motivation and justification for the MPIF algorithm, and a potential avenue for expanding Theorem~\ref{theorem:GG} to more general state space models.
In the panel model setting, the MPIF algorithm takes a similar approach to the block particle filter by updating the filtering distribution over the independent units, though the blocking occurs in the parameters space $\Theta$ rather than the latent space $\mathcal{X}$.
Thus, the MPIF algorithm has many similarities to other iterated block particle filtering algorithms that have been effective for moderately sized dynamic systems with spatial coupling \citep{ning23,ionides24}.

\section{Proof of Theorem~\ref{theorem:GG}}\label{appendix:Gaus}

As a reminder, we assume that there are $U\geq 2$ units, labeled $\seq{1}{U}$, with the data for unit $u$ being denoted as $\bm{y}^*_u$. 
We write $\theta = (\phi, \psi_1, \ldots, \psi_U)$, and recall that the unit likelihood $L_{u}(\theta;\bm{y}^*_u)$ for unit $u \in 1:U$ depends only on the shared parameter $\phi$ and the unit-specific parameter $\psi_u$.
The panel assumption implies that, conditioned on the parameter vector, the units are dynamically independent.
Therefore, we can decompose the likelihood function for the entire collection of data $L(\theta; \bm{y}^*)$ as:
\begin{align*}
    L(\theta; \bm{y}^*) = \prod_{u = 1}^U L_u(\theta;\bm{y}_u^*) = \prod_{u = 1}^U L_u(\phi, \psi_u;\bm{y}_u^*).
\end{align*}

Each iteration of the Eqs.~\ref{eq:margBayes} and \ref{eq:MPIFupdate} corresponds to a Bayes update followed by a marginalization. 
Under the statement of Theorem~\ref{theorem:GG}, the prior and likelihood are assumed to be Gaussian.
In this setting, it is well known that the resulting Bayes posterior also corresponds to a Gaussian distribution. 
Similarly, the marginalization of a multivariate Gaussian distribution also results in multivariate Gaussian distributions. 
Thus, each iteration of Eqs.~\ref{eq:margBayes} and \ref{eq:MPIFupdate} results in a density that corresponds to a Gaussian distribution. 

We now show that the mean of the resulting Gaussian distribution converges to the MLE, while the covariance matrix converges to the zero matrix, resulting in a density with all mass centered at the MLE. 
To aid this calculation, we introduce the following lemma.
\begin{lemma}\label{lemma:bound}
\label{lemma:matrix}
  Let $d$ be a positive integer, and let $B_k \in \R^{2\times 2}$ for $k \in \seq{1}{d}$ be a collection of real-valued matrices.
  We construct a sequence of matrices $A_{k} \in \R^{d+1\times d+1}$ such that:
  \begin{equation*}
    \spacingset{1.5}
    [A_k]_{i, j} = \begin{cases}
        [B_k]_{1,1}, & i = j = 1 \\ 
        [B_k]_{1,2}, & i = 1, j = k+1 \\
        [B_k]_{2, 1}, & i = k+1, j=1 \\
        [B_k]_{2,2}, & i = j = k+1 \\
        1, & i = j, \, i \notin \{1, k+1\} \\
        0, & \text{otherwise}
	\end{cases}
  \end{equation*}
 That is, $A_k$ is a perturbation of the identity matrix, where the first and $(k+1)$th row and column have been modified on the diagonal and on their off-diagonal intersection to match the matrix $B_k$.
If for all $k \in \seq{1}{d}$, $\|B_k\|_2 \leq c$ for some constant $0 < \lemmaBound \leq 1$, then
\begin{equation}\label{eq:lemma:bound}
\bigg\| \prod_{k = 1}^d A_k \bigg\|_2 \leq \lemmaBound.
\end{equation}
\end{lemma}
\begin{proof}[Proof of Lemma~\ref{lemma:bound}]
  For $i, j \in \seq{1}{(d+1)}$, denote $\mu_{(i, j)} \in \mathbb{R}^2$ as the sub-vector of $\mu \in \mathbb{R}^{d+1}$ that contains only the $i$ and $j$th elements, and write $\mu_{-(i, j)} \in \mathbb{R}^{d-1}$ to be the sub-vector of $\mu$ after removing these elements.
  By design, the matrix $A_k$ operates only on the sub-vector $\mu_{(1, k+1)}$.
That is, if $B_k \, \mu_{(1, k+1)} = (\tilde{\mu}_{(1)}, \tilde{\mu}_{(k+1)})^T$, then $A_k\mu = (\tilde{\mu}_{(1)}, \mu_{(2)}, \ldots, \tilde{\mu}_{(k + 1)}, \ldots, \mu_{(d+1)})$.
  From this, we see that for positive integer $m \leq d$,
  the first $m$ factors of the product $\prod_{k = 1}^d A_k$ modify only the first $m + 1$ dimensions of a vector $\mu \in \R^{d+1}$.
  
We proceed by mathematical induction on the dimension size.
Let $\big\{A^{(d)}_k, k\in \seq{1}{d-1}\big\}$ be a collection of matrices satisfying the condition of Lemma~\ref{lemma:bound} for each $d$.
Setting $P_d=\big\|\prod_{k = 1}^{d-1} A^{(d)}_k \big\|_2$, we first observe that Eq.~(\ref{eq:lemma:bound}) holds for $d=2$ as a direct consequence of the condition $\|B_k\|_2 \leq c$.
Suppose inductively that the lemma holds for $d$, so that $P_d\leq\lemmaBound$.
We wish to bound $P_{d+1}$.
We can choose $A^{(d)}_k$ to be the $(d+1)\times (d+1)$ sub-matrix of $A^{(d+1)}_k$ omitting row and column $(d+2)$.
Thus, for $k\in\seq{1}{d}$,
\begin{equation*}
A^{(d+1)}_k =
  \begin{pmatrix}
    A^{(d)}_k & 0 \\
    0 & 1
  \end{pmatrix}.
\end{equation*}
Consider a vector $\mu \in \R^{d+1}$, such that $\|\mu\|_2 \leq 1$.
Let $\tilde{\mu} \in \R^{d}$ be defined by 
\begin{equation}
\tilde{\mu} = \bigg[\bigg(\prod_{k = 1}^{d-1} A^{(d+1)}_k\bigg)\mu\bigg]_{(1:d)}
\end{equation}
and notice that we have
\begin{equation}\label{eq:lemma:d}
\tilde{\mu} = \bigg(\prod_{k = 1}^{d-1} A^{(d)}_k\bigg)\mu_{(1:d)}.
\end{equation}
By construction, we have 
\begin{align*}
    \bigg|\Big(\prod_{k = 1}^d A^{(d+1)}_k\Big)\mu \bigg|^2_2 &= \bigg|\Big(A^{(d+1)}_d\prod_{k = 1}^{d-1} A^{(d+1)}_k\Big)\mu \bigg|^2_2\\
    &= \big|A^{(d+1)}_d(\tilde{\mu}_{(1:d)}, \mu_{(d+1)})^T\big|^2_2 \\
    &= \big|B^{(d+1)}_d (\tilde{\mu}_1, \mu_{(d+1)})^T\big|^2_2 + \big|\tilde{\mu}_{(1:d)}\big|^2_2 - \tilde{\mu}^2_1.
\end{align*}
Because $\|B^{(d+1)}_d\|_2 \leq \lemmaBound$, $|B^{(d+1)}_d (\tilde{\mu}_1, \mu_{(d+1)})^T|^2_2 \leq \lemmaBound^2|(\tilde{\mu}_1, \mu_{(d+1)})^T|^2_2$.
Furthermore, our inductive hypothesis applied to Eq.~(\ref{eq:lemma:d}) implies that
$|{\tilde{\mu}}_{(1:d)}|^2_2 \leq \lemmaBound^2|\mu_{(1:d)}|^2_2 \leq \lemmaBound^2|\mu|^2_2$.
Therefore,
\begin{align}
    \bigg|\Big(\prod_{k = 1}^d A_k\Big)\mu\bigg|^2_2 &\leq \lemmaBound^2({\tilde{\mu}}^2_1 + \mu_{(d+1)}^2) + \lemmaBound^2(|\mu_{(1:d)}|^2_2) - {\tilde{\mu}}^2_1 \\
    &= (\lemmaBound^2-1){\tilde{\mu}}^2_1 + \lemmaBound^2 |\mu|^2_2
    \\
    \label{eq:lemma:conclusion}
    &\leq \lemmaBound^2 \vert\mu\vert^2_2,
\end{align}
with Eq.~(\ref{eq:lemma:conclusion}) implied by $\lemmaBound \leq 1$, and hence $(\lemmaBound^2-1) \leq 0$.
It follows immediately from Eq.~(\ref{eq:lemma:conclusion}) that $P_{d+1}\le c$, completing the proof.
\end{proof}

We now return to the main argument.

\begin{proof}[Proof of Theorem~\ref{theorem:GG}]
Using the transformation invariance of the MLE, we can suppose without loss of generality that the maximum of the marginal likelihood for unit $u$ is at $\phi, \psi_u = 0$.
To help ease notation, we write the covariance matrix as
$$
\Lambda^*_u = \begin{pmatrix}
    \Lambda^{(u)}_{\phi} & \Lambda_{\phi, u} \\
    \Lambda_{\phi, u} & \Lambda_{u}
\end{pmatrix}.
$$
Let the prior density $\pi_0(\theta)$ correspond to a Gaussian distribution with mean $\mu_0 \in R^{U+1}$ and precision $\Gamma_0 = \Sigma^{-1}_0 \in \R^{U+1\times U+1}$,
\begin{align*}
    \mu_0 = \begin{pmatrix}
        \mu_0^{(\phi)} \\ 
        \mu_0^{(1)} \\ 
        \vdots \\
        \mu_0^{(U)}
    \end{pmatrix}, \, \, \hspace{2mm} \Gamma_0 = \begin{pmatrix}
        \tau^{(\phi)}_0 & 0 & \ldots & 0 \\
        0 & \tau^{(1)}_0 & & \vdots \\ 
        \vdots & & \ddots & 0\\
        0 & \ldots & 0 & \tau^{(U)}_0
    \end{pmatrix}.
\end{align*}

Eqs.~\ref{eq:margBayes} and \ref{eq:MPIFupdate} contain two indices $(\nmif, \unit)$ for the intermediate density function $\pi_{\nmif, \unit}(\theta)$.
The first index ($\nmif \in \mathbb{N}$) counts the number of times data from all units has been used in the Bayes update (Eq.~\ref{eq:margBayes}), and this corresponds to the number of complete iterations of the $\nmif$ loop in Algorithm~\ref{alg:mpif}. 
The second index ($\unit \in \seq{1}{U}$) denotes the data from which unit is currently being used to update the density, and we refer to this as a sub-iteration. 
In what follows, we assume that $\pi_{\nmif, \unit - 1}(\theta) = \pi_{\nmif - 1, \Unit}(\theta)$ if $\unit = 1$, and use a similar convention for the corresponding mean and covariance that correspond to this intermediate density. 

The density after each sub-iteration of Eqs.~\ref{eq:margBayes} and \ref{eq:MPIFupdate} is Gaussian, and the marginalization step ensures that the precision matrix from previous sub-iterations is diagonal. 
Using $\mu_{\nmif, \unit-1}$ and $\Gamma_{\nmif, \unit-1}$ to denote the mean and precision after $\nmif$ complete iterations and the $(\unit-1)$th unit-iteration, we write: 
\begin{align*}
    \mu_{\nmif, \unit - 1} = \begin{pmatrix}
        \mu_{\nmif, \unit-1}^{(\phi)} \\ 
        \mu_{\nmif, \unit-1}^{(1)} \\ 
        \vdots \\
        \mu_{\nmif, \unit-1}^{(\Unit)}
    \end{pmatrix}, \, \, \hspace{2mm} \Gamma_{\nmif, \unit-1} = \begin{pmatrix}
        \tau^{(\phi)}_{\nmif, \unit-1} & 0 & \ldots & 0 \\
        0 & \tau^{(1)}_{\nmif, \unit-1} & & \vdots \\ 
        \vdots & & \ddots & 0\\
        0 & \ldots & 0 & \tau^{(\Unit)}_{\nmif, \unit-1}
    \end{pmatrix}.
\end{align*}
By design, each sub-iteration $u$ only modifies the elements of the mean and precision that correspond to the shared parameter $\phi$ and the unit specific parameter $\psi_u$. Thus, we use the superscript notation $\mu_{\nmif, \unit}^{(1, k)}$ and $\Gamma_{\nmif, \unit}^{(1, k)}$ to denote the components of the vector $\mu_{\nmif, \unit}$ corresponding to parameters $\phi$ and $\psi_k$, and the $2\times 2$ submatrix of $\Gamma_{\nmif, \unit}$ with the elements corresponding to parameters $\phi$ and $\psi_k$. 

Performing the Bayes update to this distribution (Eq.~\ref{eq:margBayes}) results in an unmarginalized precision matrix, which we denote as $\tilde{\Gamma}_{\nmif, \unit}$, where the only modified components are
$$
\tilde{\Gamma}^{(\nmif, \unit)}_{\nmif, \unit} = \Gamma^{(\nmif, \unit)}_{\nmif, \unit-1} + \Lambda_\unit^* = \begin{pmatrix}
  \tau^{(\phi)}_{\nmif, \unit-1} + \Lambda^{(\unit)}_{\phi} & \Lambda_{\phi, \unit} \\ 
  \Lambda_{\phi, \unit} & \tau^{(\unit)}_{\nmif, \unit-1} + \Lambda_{\unit}
\end{pmatrix}. 
$$
The corresponding mean $\tilde{\mu}_{(\nmif, \unit)} = \mu_{(\nmif, \unit)}$, noting that the marginalization procedure (Eq.~\ref{eq:MPIFupdate}) does not affect the mean, remains the same as the previous iteration except for the components
\begin{align}
\begin{split}
    \tilde{\mu}^{(\phi, \unit)}_{\nmif, \unit} &= \big(\tilde{\Gamma}^{(\nmif, \unit)}_{\nmif, \unit}\big)^{-1}\big(\Gamma^{(\nmif, \unit)}_{\nmif, \unit-1} \mu_{\nmif, \unit-1}^{(\phi, u)} + \Lambda_u^* (0, 0)^T\big) \\
   \mu^{(\phi, u)}_{\nmif, \unit} &= \big(\tilde{\Gamma}^{(\phi, u)}_{\nmif, \unit}\big)^{-1}\Gamma^{(\phi, u)}_{\nmif, \unit-1} \mu_{\nmif, \unit-1}^{(\phi, u)} \\
   &= B_{\nmif, \unit} \mu_{\nmif, \unit-1}^{(\phi, u)}, 
\end{split}\label{eq:submatB}
\end{align}
where $B_{\nmif, \unit} = \big(\tilde{\Gamma}^{(\phi, 1)}_{\nmif, \unit}\big)^{-1}\Gamma^{(\phi, u)}_{\nmif, \unit-1}$ is a $2\times 2$ matrix that provides the update to the first and $(\unit+1)$th components of a vector $\mu \in \mathbb{R}^{\Unit+1}$.
Now by defining matrix $A_{\nmif, \unit}$ to be a perturbation of the $\Unit+1$ identity matrix, such that
\begin{equation*}
\spacingset{1.5}
\big[A_{\nmif, \unit}\big]_{i, j} = \begin{cases}
  [B_{\nmif, \unit}]_{1, 1} & i = j = 1 \\
  [B_{\nmif, \unit}]_{2, 2} & i = j = u+1 \\
  [B_{\nmif, \unit}]_{1, 2} & i = 1, j = k+1 \\
  [B_{\nmif, \unit}]_{2, 1} & i = k+1, j = 1 \\
  1 & i = j, i \notin \{1, k+1\} \\
  0 & \text{otherwise}
\end{cases},
\end{equation*}
We see that for a given precision matrix $\Gamma_{\nmif, \unit-1}$, the a single sub-iteration of Eqs.~\ref{eq:margBayes} and \ref{eq:MPIFupdate} corresponds to a linear update of the mean vector $\mu_{\nmif, \unit-1}$, defined by
\begin{align*}
  \mu_{\nmif, \unit} = A_{\nmif, \unit} \, \mu_{\nmif, \unit-1}.
\end{align*}
Thus, the mean after $M$ complete iterations, denoted by $\mu_{M}$ is calculated by the product of these linear transformations
\begin{align*}
  \mu_M = \bigg(\prod_{\nmif = 1}^M\prod_{u = 1}^U A_{\nmif, \unit}\bigg)\mu_0. 
\end{align*}
Thus, the long-term value of $\mu_M$ depends on the long-term behavior of the product of matrices $A_{\nmif, \unit}$. 

We now consider the limiting behavior of the precision matrices, which determines the behavior of the matrices $A_{\nmif, \unit}$. 
First, we recall that the unmarginalized precision matrix is of the form:
\begin{align*}
  {\tilde{\Gamma}}^{(\nmif, \unit)}_{\nmif, \unit} = \begin{pmatrix}
  \tau^{(\phi)}_{\nmif, \unit-1} + \Lambda^{(u)}_{\phi} & \Lambda_{\phi, u} \\ 
  \Lambda_{\phi, u} & \tau^{(u)}_{\nmif, \unit-1} + \Lambda_{u}\end{pmatrix}. 
\end{align*}
The marginalization step (Eq.~\ref{eq:MPIFupdate}) for Gaussian densities is represented by setting the off-diagonal elements of the covariance matrix to be zero.
Thus, if we generalize this step by writing
\begingroup
\renewcommand{\arraystretch}{0.7}
$
{\tilde{\Gamma}}^{(\nmif, \unit)}_{\nmif, \unit} = \begin{pmatrix}a & b \\ b & d \end{pmatrix},
$
\endgroup then
\begingroup
\renewcommand{\arraystretch}{0.7}
$
\big({\tilde{\Gamma}}^{(\nmif, \unit)}_{\nmif, \unit}\big)^{-1} = \frac{1}{ad - b^2}\begin{pmatrix}d & -b \\ -b & a\end{pmatrix}
$
\endgroup
and the marginalized version of the matrix is
\begingroup
\renewcommand{\arraystretch}{0.7}
$
\big(\Gamma^{(\nmif, \unit)}_{\nmif, \unit}\big)^{-1} = 
    \begin{pmatrix} \frac{d}{ad-b^2} & 0 \\ 0 & \frac{a}{ad-b^2}\end{pmatrix};
$
\endgroup
by taking the inverse of this matrix, we get the precision matrix after a unit-update to be
\begingroup
\renewcommand{\arraystretch}{0.7}
$
   \Gamma^{(\nmif, \unit)}_{\nmif, \unit} = \begin{pmatrix} a - \frac{b^2}{d} & 0 \\ 0 & d-\frac{b^2}{a}\end{pmatrix}.
$
\endgroup

Using this general calculation, the precision matrix after the marginalization step is 
\begin{align}
 \Gamma^{(\phi, u)}_{\nmif, \unit} &= \begin{pmatrix}
    \tau_{\nmif, \unit-1}^{(\phi)} + \Lambda_{\phi}^{(u)} - \frac{\Lambda_{\phi, u}^2}{\tau_{\nmif, \unit-1}^{(u)} + \Lambda_u} & 0 \\
    0 & \tau_{\nmif, \unit-1}^{(u)} + \Lambda_u - \frac{\Lambda_{\phi, u}^2}{\tau_{\nmif, \unit-1}^{(\phi)} + \Lambda^{(u)}_{\phi}}
   \end{pmatrix} \nonumber \\
   &= \begin{pmatrix}
    \tau_{\nmif, \unit-1}^{(\phi)} + \Lambda_{\phi}^{(u)} - \alpha_{\nmif, \unit} & 0 \\
    0 & \tau_{\nmif, \unit-1}^{(u)} + \Lambda_u - \beta_{\nmif, \unit}
   \end{pmatrix},\label{eq:margPrecision}
\end{align}
with $\alpha_{\nmif, \unit} = \frac{\Lambda_{\phi, u}^2}{\tau_{\nmif, \unit-1}^{(u)} + \Lambda_u}$ and $\beta_{\nmif, \unit} = \frac{\Lambda_{\phi, u}^2}{\tau_{\nmif, \unit-1}^{(\phi)} + \Lambda^{(u)}_{\phi}}$. Because $\Lambda^*_u$ is a positive definite matrix, for all real numbers $c > 0$, 
\begin{align*}
    \Lambda_{\phi}^{(u)} > \frac{\Lambda^2_{\phi, u}}{\Lambda_{u}} > \frac{\Lambda^2_{\phi, u}}{c + \Lambda_{u}} 
  \end{align*}
And therefore by letting $\alpha = \min_u \big(\Lambda_u - \frac{\Lambda^2_{\phi, u}}{\Lambda^{(u)}_{\phi}}\big) > 0$, we have $\tau_{\nmif, \unit}^{(\phi)} > \tau_{\nmif, \unit-1} + \alpha$ for all $\nmif \in \mathbb{N}$ and $\unit \in \seq{1}{\Unit}$. 
By iterating this inequality, we have $\tau_{\nmif, \unit}^{(\phi)} > \tau^{(\phi)}_0 + \nmif\alpha$, and thus we see that $\tau_{\nmif, \unit}^{(\phi)} = O(\nmif)$. 
Therefore as $\nmif \rightarrow \infty$, $\tau_{\nmif, \unit}^{(\phi)} \rightarrow \infty$. 
A similar calculation shows that $\tau_{\nmif, k}^{(\unit)} = O(\nmif)$ for all $k, \unit \in \seq{1}{\Unit}$. 
Thus, the covariance matrix after $\nmif$ complete iterations has zeros on off-diagonal elements, and the diagonal elements are of order $O(1/\nmif)$, proving the statement in Theorem~\ref{theorem:GG} that $\|\Sigma_\nmif\|_2 \rightarrow 0$. 
To finish the proof, we need to show that $|\mu_\nmif|_2 \rightarrow 0$. 
To do this, we need more precise descriptions for the rates at which the precision grows. 

Our previous calculations establish that both $\tau^{(\phi)}_{\nmif, k}$ and $\tau^{(\unit)}_{\nmif, k}$ are strictly increasing sequences in both $\nmif$ and $k$, and are unbounded. 
Thus, the correction terms $\alpha_{\nmif, \unit}$ and $\beta_{\nmif, \unit}$ converge to zero, as the only moving parts are the precision terms that are in the denominator of each of these sequences. 
Furthermore, the sequence $\{\alpha_i\}$, defined as
\begin{align*}
  \alpha_i = \sum_{u = 1}^U \alpha_{i, u} 
\end{align*}
satisfies $\alpha_i \rightarrow 0$. 
Thus, by iterating Eq.~\ref{eq:margPrecision}, we can express the precision matrix after $\nmif$ complete iterations as
\begin{align*}
  \Gamma^{(\phi, u)}_{\nmif} = \Gamma^{(\phi, \unit)}_{\nmif-1, \Unit} = \begin{pmatrix}
    \tau^{(\phi)}_{0} + \nmif\sum_{k = 1}^{\Unit}\Lambda^{(k)}_{\phi} - \sum_{i = 1}^\nmif\alpha_{i} & 0 \\
    0 & \tau^{(\unit)}_{0} + \nmif\Lambda_{k} - \sum_{i = 1}^\nmif \beta_{i, \unit}
  \end{pmatrix}.
\end{align*}
Using the fact that the Ces\`aro mean of a convergent sequence converges to the same limit as the sequence, we have $\frac{1}{\nmif}\sum_{i = 1}^\nmif \alpha_{i} \rightarrow 0$, and therefore 
\begin{align}
  \Gamma^{(\phi, k)}_{\nmif, \unit} &= \begin{pmatrix} 
    \nmif\sum_{k = 1}^\Unit \Lambda_{\phi}^{(k)} + o(\nmif) & 0 \\
    0 & \nmif\Lambda_k + o(\nmif)
  \end{pmatrix}. \label{eq:nGamma}
\end{align}

Using the rates established in Eq.~\ref{eq:nGamma}, we now investigate the long-term behavior of the mean vector $\mu_\nmif$ by considering the spectral norm of the matrices $B_{\nmif, \unit}$ that define the linear updates to the sub-vector $\mu^{(\phi, \unit)}_{\nmif}$.
Recall from Eq.~\ref{eq:submatB} that $B_{\nmif, \unit} = \big(\Gamma^{(\phi, u)}_{\nmif, \unit-1} + \Lambda^*_{u}\big)^{-1}\Gamma^{(\phi, u)}_{\nmif, \unit-1}$. 
We can take advantage of the fact that $\Gamma^{(\phi, u)}_{\nmif, \unit-1}$ is a diagonal matrix to write
\begin{align*}
  B_{\nmif, \unit} &= \big(\Gamma^{(\phi, u)}_{\nmif, \unit-1} + \Lambda^*_{u}\big)^{-1}\Gamma^{(\phi, u)}_{\nmif, \unit-1} \\
  &= \Big(I + (\Gamma^{(\phi, u)}_{\nmif, \unit-1})^{-1}\Lambda^*_{u}\Big)^{-1} \\
  &= \begin{pmatrix} 1 + \frac{\Lambda^{(u)}_{\phi}}{\nmif\sum_{k = 1}^U\Lambda^{(k)}_{\phi} + o(\nmif)} & \frac{\Lambda_{\phi, u}}{\nmif\sum_{k = 1}^U\Lambda^{(k)}_{\phi} + o(\nmif)} \\ 
  \frac{\Lambda_{\phi, u}}{\nmif\Lambda_u + o(\nmif)} & 1 + \frac{\Lambda_{u}}{\nmif\Lambda_u + o(\nmif)} \end{pmatrix}^{-1} \\
  &= \begin{pmatrix} 1 + \frac{\Lambda^{(u)}_{\phi}}{\nmif\sum_{k = 1}^U\Lambda^{(k)}_{\phi}} + o(1/\nmif) & \frac{\Lambda_{\phi, u}}{\nmif\sum_{k = 1}^U\Lambda^{(k)}_{\phi}} + o(1/\nmif)\\ 
  \frac{\Lambda_{\phi, u}}{\nmif\Lambda_u} + o(1/\nmif) & 1 + \frac{1}{\nmif} + o(1/\nmif) \end{pmatrix}^{-1}.
\end{align*}
Next we would like to calculate the spectral norm of $B_{\nmif, \unit}$. 
Recall that for an invertible matrix $A$, $\|A^{-1}\|_2 = \sigma_{\max}(A^{-1}) = 1 / \sigma_{\min}(A)$, where $\sigma_{\max}$ and $\sigma_{\min}$ correspond to the maximum and minimum singular values of $A$.
Therefore in order to calculate $\|B_{\nmif, \unit}\|_2$, we need to find the minimum singular value of $B^{-1}_{\nmif, \unit}$, which is equal to the minimum eigenvalue of the matrix $B^{-T}_{\nmif, \unit}B^{-1}_{\nmif, \unit}$:
\begin{align*}
  B^{-T}_{\nmif, \unit}B^{-1}_{\nmif, \unit} &= \begin{pmatrix} 1 + \frac{2\Lambda^{(u)}_{\phi}}{\nmif\sum_{k = 1}^U\Lambda^{(k)}_{\phi}} + o(1/\nmif) & \frac{\Lambda_{\phi, u}}{\nmif} \big(\frac{1}{\sum_{k = 1}^U\Lambda^{(k)}_{\phi}} + \frac{1}{\Lambda_u}\big) + o(1/\nmif)\\ 
  \frac{\Lambda_{\phi, u}}{\nmif} \big(\frac{1}{\sum_{k = 1}^U\Lambda^{(k)}_{\phi}} + \frac{1}{\Lambda_u}\big) + o(1/\nmif) & 1 + \frac{2}{\nmif} + o(1/\nmif) \end{pmatrix} \\
  &= I + \frac{1}{\nmif}\begin{pmatrix} \frac{2\Lambda^{(u)}_{\phi}}{\sum_{k = 1}^U\Lambda^{(k)}_{\phi}} + o(1) & \Lambda_{\phi, u} \big(\frac{1}{\sum_{k = 1}^U\Lambda^{(k)}_{\phi}} + \frac{1}{\Lambda_u}\big) + o(1)\\ 
  \Lambda_{\phi, u} \big(\frac{1}{\sum_{k = 1}^U\Lambda^{(k)}_{\phi}} + \frac{1}{\Lambda_u}\big) + o(1) & 2 + o(1) \end{pmatrix} \\
  &= I + \frac{1}{\nmif}C_{\nmif, \unit}.
\end{align*}
Using this expression, the eigenvalues of $B^{-T}_{\nmif, \unit}B^{-1}_{\nmif, \unit}$ are equal to one plus the eigenvalues of $\frac{1}{\nmif}C_{\nmif, \unit}$. 
Thus, we consider the characteristic polynomial defined by $\det (C_{\nmif, \unit} - \lambda I)$ to calculate the eigenvalues of $C_{\nmif, \unit}$. The resulting polynomial is
\begin{align*}
f_{\nmif, \unit}(\lambda) &=
  \bigg(
    \frac{2\Lambda^{(u)}_{\phi}}{\sum_{k = 1}^U\Lambda^{(k)}_{\phi}} - \lambda\bigg)
  \big(2 - \lambda\big) -
  \Lambda_{\phi, u}^2
  \bigg(
    \frac{1}{\sum_{k = 1}^U\Lambda^{(k)}_{\phi}} + \frac{1}{\Lambda_u}
  \bigg)^2 + o(1)
  \\
  &= \lambda^2 -
  2\bigg(
    1 + \frac{\Lambda^{(u)}_{\phi}}{\sum_{k = 1}^U\Lambda^{(k)}_{\phi}}
  \bigg)
  \lambda +
  \bigg\{
    \frac{4\Lambda^{(u)}_{\phi}}{\sum_{k = 1}^U\Lambda^{(k)}_{\phi}}
    - \Lambda_{\phi, u}^2
    \bigg(
      \frac{1}{\sum_{k = 1}^U\Lambda^{(k)}_{\phi}} + \frac{1}{\Lambda_u}
    \bigg)^2
  \bigg\} + o(1).
\end{align*}
We now proceed by showing that, under the condition specified in Assumption~\ref{eq:GausAssumption}, all roots of the polynomial $f_{\nmif, \unit}(\lambda)$ are strictly positive. 
That is, for some $0 < \lambda^{(1)}_{\nmif, \unit} < \lambda^{(2)}_{\nmif, \unit}$, we have the eigenvalues of $C_{\nmif, \unit}$ are $\lambda^{(1)}_{\nmif, \unit} + o(1)$ and $\lambda^{(2)}_{\nmif, \unit} + o(1)$. 

For this to hold, we need first that the discriminant of the resulting quadratic equation to be strictly positive. 
This condition holds without any assumptions, as we can write the condition as
\begin{align*}
\bigg(2\bigg[ 1 + \frac{\Lambda^{(u)}_{\phi}}{\sum_{k = 1}^U\Lambda^{(k)}_{\phi}}\bigg]\bigg)^2 - &4\bigg(\frac{4\Lambda^{(u)}_{\phi}}{\sum_{k = 1}^U\Lambda^{(k)}_{\phi}} - \Lambda_{\phi, u}^2\bigg[\frac{1}{\sum_{k = 1}^U\Lambda^{(k)}_{\phi}} + \frac{1}{\Lambda_u}\bigg]^2\bigg) > 0 \iff \\
\Lambda^{2}_{\phi, u}\bigg(&\frac{1}{\sum_{k = 1}^U \Lambda_{\phi}^{(k)}} + \frac{1}{\Lambda_u}\bigg)^2 > -\frac{\big(\sum_{k = 1}^U \Lambda_{\phi}^{(k)} - \Lambda_\phi^{(u)}\big)^2}{\big(\sum_{k = 1}^U \Lambda_{\phi}^{(k)}\big)}.
\end{align*}
This condition always holds as the left-hand side of the inequality is strictly positive, and the right-hand side of the inequality is strictly negative. 
Thus, the polynomial $f_{\nmif, \unit}(\lambda)$ has two unique real roots. 

Next, by taking the derivative and setting equal to zero, we note that $\argmin_\lambda f_{\nmif, \unit}(\lambda) = \lambda^* > 0$, as the coefficient on the linear term of the polynomial is always negative.
This combined with the previous condition ensures that $\lambda^{(2)}_{\nmif, \unit} > 0$. 

Finally, if $f_{\nmif, \unit}(0) > 0$, then the intermediate value theorem implies that because $f_{\nmif, \unit}(\lambda^*) < 0$, then there exists some $\lambda^{(1)}_{\nmif, \unit} \in (0, \lambda^{*})$ such that $f_{\nmif, \unit}(\lambda^{(1)}_{\nmif, \unit}) = 0$, implying that $\lambda^{(1)}_{\nmif, \unit}$ is a positive root.
Therefore we need 
$$
\frac{4\Lambda^{(u)}_{\phi}}{\sum_{k = 1}^U\Lambda^{(k)}_{\phi}} - \Lambda_{\phi, u}^2\bigg(\frac{1}{\sum_{k = 1}^U\Lambda^{(k)}_{\phi}} + \frac{1}{\Lambda_u}\bigg)^2 > 0,
$$
which is equivalent to the condition
\begin{align*}
  \Lambda_{\phi, u}^2 < \frac{4\Lambda^{(u)}_\phi\Lambda^2_{u}\sum_{k = 1}^U \Lambda_{\phi}^{(k)}}{\big(\Lambda_u + \sum_{k = 1}^{U}\Lambda^{(k)}_{\phi}\big)^2}.
\end{align*}
This condition is guaranteed by Assumption~\ref{eq:GausAssumption}, ensuring that the eigenvalues of $C_{\nmif, \unit}$ are $\lambda^{(1)}_{\nmif, \unit} + o(1)$ and $\lambda^{(1)}_{\nmif, \unit} + o(1)$ for two positive numbers $\lambda^{(1)}_{\nmif, \unit}$ and $\lambda^{(1)}_{\nmif, \unit}$.
As a result of this computation, The minimum eigenvalue of $B^{-T}_{\nmif, \unit}B^{-1}_{\nmif, \unit}$ is equal to $1 + \frac{\lambda^{(1)}_{\nmif, \unit}}{\nmif} + o(1/\nmif)$, implying that 
\begin{align*}
  \big\| B_{\nmif, \unit}\big\|_2 &= 1 / \sigma_{min}\big(B^{-1} \big) \\
  &= \frac{1}{1 + 1 + \frac{\lambda^{(1)}_{\nmif, \unit}}{\nmif} + o(1/\nmif)} \\
  &= 1 - \frac{\lambda^{(1)}_{\nmif, \unit}}{\nmif} + o(1/\nmif). 
\end{align*}
By Lemma~\ref{lemma:matrix}, we have for all $\nmif \in \mathbb{N}$, $\|\prod_{u = 1}^U A_{\nmif, \unit}\|_2 \leq 1 - \frac{\max_u \lambda^{(1)}_{\nmif, \unit}}{\nmif} + o(1/\nmif)$, and therefore by the sub-multiplicative property of the spectral norm,
\begin{align*}
  \Big\| \prod_{i = 1}^m \prod_{u = 1}^U A_{i, u}\Big\|_2 & \leq \prod_{i = 1}^m \Big\|\prod_{u = 1}^U A_{i, u}\Big\|_2 \\
  &< \prod_{i = 1}^m \Big(1 - \frac{\max_u \lambda^{(1)}_{i, \unit}}{i} + o(1/i)\Big) \rightarrow 0.
\end{align*}
Therefore for any arbitrary initial conditions $\mu_0$ and $\Gamma_0$, we have
$$|\mu_m|_2 = \bigg|\bigg(\prod_{i = 1}^m\prod_{u = 1}^U A_{i, \unit}\bigg)\mu_0\bigg|_2 \rightarrow 0$$
as $m\rightarrow \infty$, completing the proof.
\end{proof}

  \subsection{Proof of Corollary~\ref{corollary:perturbed}}\label{appendix:perturbed} 
 
  For this corollary, we use the same setup as Appendix~\ref{appendix:Gaus}. 
  At each step, we add independent random perturbations to the current parameter distribution via a convolution operation before updating using Bayes rule. 
  Because Gaussian density convolved with another Gaussian density is still Gaussian, we again only record the mean and precision corresponding to the resulting Gaussian distribution at each step, which we denote $\mu'_{\nmif} \in \R^{U+1}$ and $\Gamma'_{\nmif} \in \R^{U+1\times U+1}$. 

\begin{proof}
Let $\mu_{\nmif, \unit} \in \R^{U+1}$ and
$\Gamma_{\nmif, \unit} = \text{diag}
  \big(
    \tau_{\nmif, \unit}^{(\phi)}, \tau_{(\nmif, \unit)}^{(1)}, \ldots, \tau_{(\nmif, \unit)}^{(U)}
  \big)
  \in \R^{U+1\times U+1}$
denote the mean and precision matrices of the multivariate Gaussian after the $(\nmif, \unit)$th iteration. 
Consider the $(\nmif, \unit)$th update of Eqs.~\ref{eq:margBayesPerturb}--\ref{eq:MPIFupdatePerturb}. 
Using well-known results on the convolution of Gaussian densities and the conjugate prior identities, the umarginalized precision is evaluated as
\begin{align*}
\tilde{\Gamma}'_{\nmif, \unit+1} &= \begin{pmatrix} \frac{\tau_{\nmif, \unit}^{(\phi)}\tau_{\nmif, \unit}^{(\sigma)}}{\tau_{\nmif, \unit}^{(\phi)} + \tau_{\nmif, \unit}^{(\sigma)}} + \Lambda^{(u)}_\phi & \Lambda_{\phi, u} \\ \Lambda_{\phi, u} & \frac{\tau_{\nmif, \unit}^{(u)}\tau_{\nmif, \unit}^{(\sigma)}}{\tau_{\nmif, \unit}^{(u)} + \tau_{\nmif, \unit}^{(\sigma)}} + \Lambda_u\end{pmatrix}.
   \end{align*}
   Following the formula from Appendix~\ref{appendix:Gaus} for finding the marginalized precision matrices, the marginalized precision $\Gamma'_{\nmif, \unit}$ can be expressed as
     \begin{align*}
   \Gamma'_{\nmif, \unit+1} &= \begin{pmatrix} \frac{\tau^{(\phi)}_\nmif\tau^{(\sigma)}_\nmif}{\tau^{(\phi)}_\nmif + \tau^{(\sigma)}_\nmif} + \Lambda^{(\unit)}_{\phi} - \frac{\Lambda^2_{\phi, u}(\tau^{(\sigma)}_\nmif + \tau^{(u)}_\nmif)}{\tau^{(\sigma)}_\nmif\tau^{(\unit)}_\nmif + \Lambda_u(\tau^{(\sigma)}_\nmif + \tau^{(u)}_\nmif)} & 0 \\ 0 & \frac{\tau^{(u)}_\nmif\tau^{(\sigma)}_\nmif}{\tau^{(u)}_\nmif + \tau^{(\sigma)}_\nmif} + \Lambda_{u} - \frac{\Lambda^2_{\phi, u}(\tau^{(\sigma)}_\nmif + \tau^{(\phi)}_\nmif)}{\tau^{(\sigma)}_\nmif\tau^{(\phi)}_\nmif + \Lambda^{(u)}_\phi(\tau^{(\sigma)}_\nmif + \tau^{(\phi)}_\nmif)}\end{pmatrix}.
   \end{align*}
  Because $\Lambda^*_u$ is positive definite, the same argument in the unperturbed case (Appendix~\ref{appendix:Gaus}) gives the existence of positive constants $\{\beta_u, \alpha_u\}_{u = 1}^U$ such that
  \begin{align*}
    \tau^{(\phi)}_{\nmif, \unit+1} &> \frac{\tau^{(\phi)}_\nmif\tau^{(\sigma)}_\nmif}{\tau^{(\phi)}_\nmif + \tau^{(\sigma)}_\nmif} + \alpha_u \\
    \tau^{(u)}_{\nmif, \unit+1} &> \frac{\tau^{(u)}_\nmif\tau^{(\sigma)}_\nmif}{\tau^{(u)}_\nmif + \tau^{(\sigma)}_\nmif} + \beta_u.
  \end{align*}
  Furthermore, the sequence defined by iterating the right hand side of these inequalities is unbounded, which is easily demonstrated by assuming it is bounded, and noting that this leads to a contradiction because $1/\tau^{(\sigma)}_\nmif = o(1/\nmif)$, and therefore each update grows by an amount arbitrarily close to the constants $\alpha_u$ and $\beta_u$.
  
  Immediately, this result implies that $\tau_{\nmif, \unit}^{(\phi)}, \tau_{\nmif, \unit}^{(u)} \rightarrow \infty$ as $\nmif \rightarrow \infty$. 
  That is, the covariance matrix convergence to the zero matrix, even in the presence of parameter perturbations. 
  In particular, we can write the correction term
  $$
  \epsilon_{\nmif, \unit} = \frac{\Lambda^2_{\phi, u}(\tau^{(\sigma)}_\nmif + \tau^{(u)}_\nmif)}{\tau^{(\sigma)}_\nmif\tau^{(u)}_\nmif + \Lambda_u(\tau^{(\sigma)}_\nmif + \tau^{(u)}_\nmif)} = o(1),
  $$
  and, using the fact that the sequence $\tau^{(\sigma)}_\nmif$ is defined independently of $\tau^{\phi}_{\nmif, \unit}$, we can write 
  \begin{align*}
  \frac{\tau^{(\phi)}_{\nmif, \unit}\tau^{(\sigma)}_\nmif}{\tau^{(\phi)}_{\nmif, \unit} + \tau^{(\sigma)}_\nmif} &= \frac{\tau^{(\phi)}_{\nmif, \unit}}{\frac{1}{\tau^{(\sigma)}_{\nmif}}\big(\tau^{(\phi)}_{\nmif, \unit} + \tau^{(\sigma)}_{\nmif}\big)} \\
  &= \tau_{\nmif, \unit}^{(\phi)}\frac{1}{1+o(1/\nmif)} \\
  &= \tau^{(\phi)}_{\nmif, \unit} + o(1/\nmif).
  \end{align*}
  Thus, the updated precision after a full iteration can be approximated as
  $$
  \tau_{\nmif + 1}^{(\phi)} = \tau^{\phi}_{\nmif} + \nmif\sum_{u = 1}^U \Lambda_{\phi}^{(u)} - \sum_{u = 1}^U\epsilon_{\nmif, \unit} + o(1/\nmif),
  $$
  and therefore 
  $$
  \tau_{\nmif + 1}^{(\phi)} = o(\nmif) + \nmif\sum_{u = 1}^U \Lambda_{\phi}^{(u)}, 
  $$
  and the same basic calculation shows that for all $u$, 
  $$
  \tau^{(u)}_{\nmif+1} = o(\nmif) + \nmif\Lambda_{u}.
  $$
  
Similar to the unperturbed case, the update to the mean can be expressed as $\mu_{\nmif+1} = \big(\prod_{u = 1}^UA'_{\nmif, \unit}\big)\mu_\nmif$, where the matrix $A'_{\nmif, \unit}$ is defined in the same way as $A_{\nmif, \unit}$ from Appendix~\ref{appendix:Gaus}, after replacing the covariance matrices with the perturbed versions: 
$$
A'_{\nmif, \unit} = \big(\Gamma'_{\nmif, \unit - 1} + \Lambda^*_u\big)^{-1}\Gamma'_{\nmif, \unit - 1}.
$$
We have shown above that $A'_{\nmif, \unit}$ and $A_{\nmif, \unit}$ are asymptotically equivalent. 
Therefore the same argument in the proof of Theorem~\ref{theorem:GG} immediately implies
$$
\big| \mu'_{\nmif} \big|_2 =
  \bigg|
    \bigg(
      \prod_{n = 1}^\nmif\prod_{u = 1}^U A'_{n, u}
    \bigg)
    \mu_0
  \bigg|_2 \rightarrow 0.
$$

   \end{proof}

\end{appendices}

\end{document}